\documentclass[11pt]{report}
\usepackage[utf8]{inputenc}
\usepackage{graphicx}
\graphicspath{ {images/} }
\usepackage{caption}
\usepackage{subcaption}
\usepackage{float}
\usepackage[width=150mm,top=35mm,bottom=25mm,bindingoffset=6mm]{geometry}
\usepackage{fancyhdr}
\usepackage{setspace}
\newtheorem{defi}{Definition}[section]
\newtheorem{cor}{Corollary}[section]

\newtheorem{lem}{Lemma}[section]
\newtheorem{prop}{Proposition}[section]
\newtheorem{obs}{Observation}[section]
\newtheorem{ex}{Example}[section]
\newcommand{\ket}[1]{|#1\rangle}
\newcommand{\bra}[1]{\langle #1|}      
\newcommand{\braket}[2]{\langle #1 | #2 \rangle}  
\usepackage{framed}
\usepackage{lipsum}
\usepackage{hyperref}
\usepackage{xcolor}
\usepackage{enumerate}
\usepackage{multicol}
\usepackage{mathtools}
\usepackage{amsfonts}
\usepackage{amsmath}
\usepackage{tikz}
\usetikzlibrary{quantikz2}
\newenvironment{proof}{\noindent \textbf{Proof:}}{\hfill$\square$}
\renewcommand{\headrulewidth}{0pt}
\usepackage[numbers]{natbib}

\let\oldboldsymbol\boldsymbol
\renewcommand{\boldsymbol}[1]{\ifmmode\oldboldsymbol{#1}\else\texorpdfstring{$\oldboldsymbol{#1}$}{#1}\fi}

\title{Thesis Title}
\author{Author Name}
\date{Day Month Year}

\begin{document}

\begin{titlepage}
    \begin{center}
        \vspace*{1cm}
        
        \Huge
        \textbf{Designing tight frames for quantum computing}
        
        \vspace{0.5cm}
        \LARGE

        \textbf{Luis Quezada}
\vfill

        \begin{figure}[hbtp]
\centering
\includegraphics[scale=0.255]{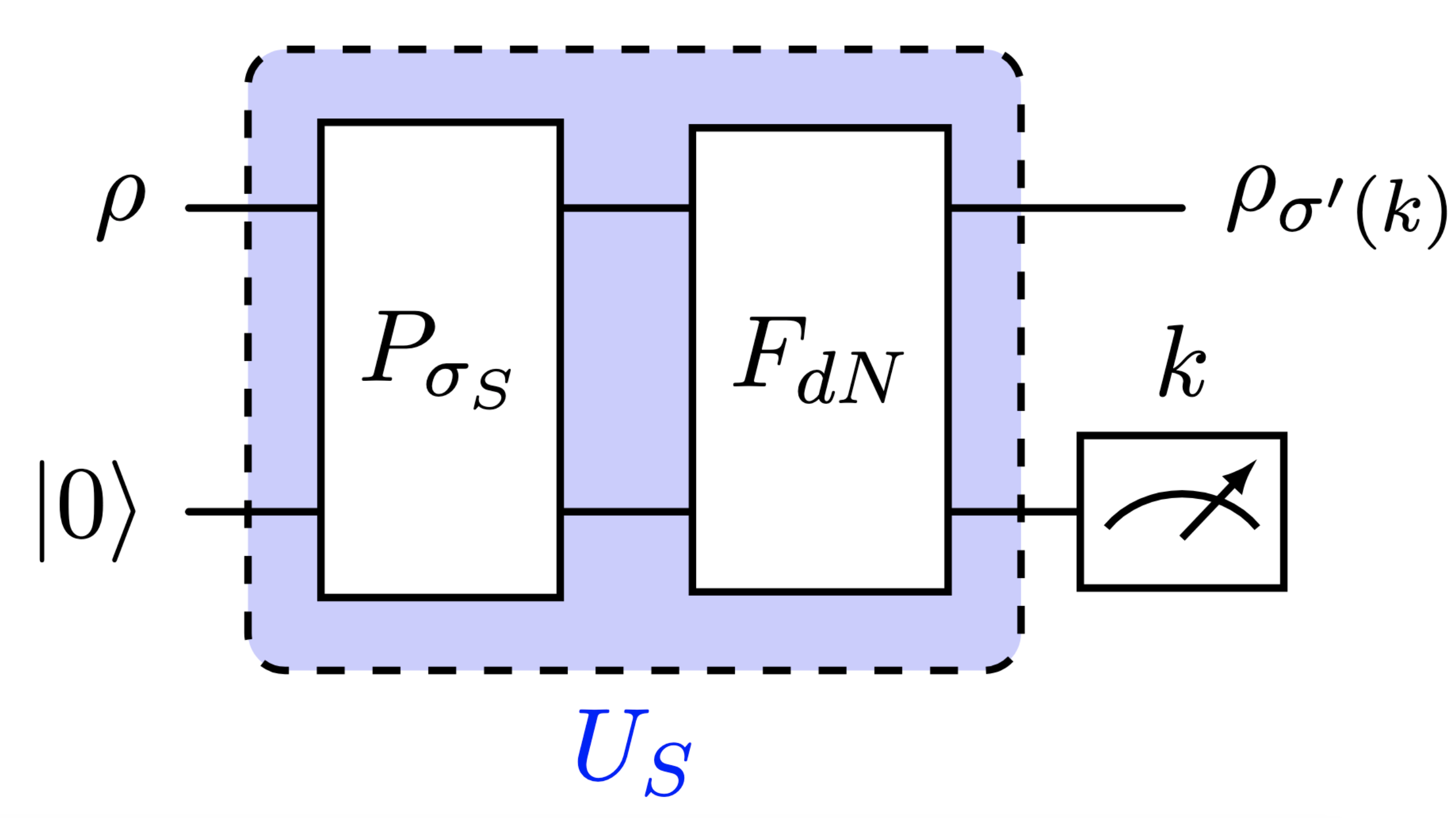}
\end{figure}

        \vfill
        
        A dissertation submitted in partial fulfillment\\
        of the requirements for the degree of\\
        Physics graduate\\

        \vspace{1.8cm}

        \Large
        At the\\Pontificia Universidad Católica de Chile\\
        Year 2024\\ 
        \vspace{1.0cm}
        \begin{flushleft}
        \large
        Date of Final Oral Exam: 24/01/2025\\
        Advising Professor: Dardo Goyeneche, Physics\\
        Committee: \\
        \setlength{\parindent}{10ex}
        Aldo Delgado Hidalgo, Professor, Physics\\Giuseppe de Nittis,  Professor, Mathematics\\ 
        \end{flushleft}
        
    \end{center}
    
\end{titlepage}

\chapter*{}
\pagenumbering{roman}

\clearpage

\doublespacing

\vspace*{\fill} 
\begin{center}
    {\Huge \textbf{Dedication}} 
\end{center}
\vspace{1.5cm} 

\begin{center}
    {\large To my family, who stood as my pillar during difficult times. \\
    To my friends, with whom I shared the enthusiasm for my work. \\
    To my teachers, who guided me and kept me focused on my goals. \\
    And to all those who, even with little knowledge of physics or mathematics, inspired me to undertake this project.}
\end{center}
\vspace*{\fill} 

\chapter*{Acknowledgements}

I would like to begin by expressing my deepest gratitude to my advisor, Dardo Goyeneche, for his invaluable patience. His guidance was essential for this project, as was his insistence on addressing key topics. Despite initial disagreements, he managed to guide me while granting me the freedom to work in my own style. I am also grateful to my dear friend Herbert Díaz, with whom I could share all kinds of physics-related "madness" without fear of judgment. To my grandmother, Blanca Rosa, who was there every day, in good times and bad, providing me with the affection needed to keep moving forward. To my mother, whose advice continuously guides my life toward a better future. To the Pontificia Universidad Católica de Chile, which over these years has become a home to me, and which I sincerely hope will continue to be so. Lastly, I am deeply thankful to the Fondecyt Regular Project Nro.1230586 for its financial support, which was crucial for the development of this work.

\newpage

\fancyhf{} 
\fancyhead[RO,R]{\thepage} 
\renewcommand{\headrulewidth}{0pt}

\begin{center}
    {\Huge \textbf{Abstract}} 
\end{center}
 \vspace{0.5cm}
The aim of this thesis is to explore the implementation a special kind of quantum measurements, so called Harmonic tight frames. To achieve this goal, representation theory is addressed, emphasizing its construction from irreducible representations and the relations they satisfy. These tools simplifies the study due to the existence of symmetries. Leading to the concept of group frames, defined as orbits under the unitary action of a group and characterized by their symmetry groups. Among the various groups, the simplest are the abelian ones, giving rise to harmonic frames. These are analyzed through the classification theorem of abelian groups and the characters of their irreducible representations.\\

In the quantum mechanics context, the frame elements can be interpreted as pure states of a system. Therefore, the necessary and sufficient conditions for separability are explored, as separable states are easier to implement due to their local nature. Alternatively, when viewing frames as measurements, the concept of POVMs and Naimark’s theorem are studied as key tools for designing measurements in quantum computers.\\

Using this knowledge, a characterization of the harmonic frames was given from the abelian group structure theorem, which allowed obtaining a necessary and sufficient condition for their separability. The conditions of maximum entanglement of these states for bipartite systems were also studied, obtaining a necessary condition on the dimensions of the subsystems. Finally, quantum cirtuit were obtained that allows implementing the harmonic frames associated to cyclic groups as POVM using a Fourier matrix and a permutation matrix. We conclude by giving simple examples where the results are applied to quantum computers and proposing future avenues of research that could serve to improve the circuit design and extend it to the case of all harmonics.

\singlespacing

\tableofcontents

\doublespacing

\newpage

\pagenumbering{arabic}
\setcounter{page}{1} 

\chapter{Introduction}

\bibliographystyle{IEEEtran}
Frame theory was first introduced in the 1950s by Duffin and Schaeffer in their paper "A class of nonharmonic Fourier series" \cite{duffin1952class}, which builds upon the concepts studied by D. Gabor in \cite{gabor1946theory} regarding signal representation. Frames offer significant flexibility in decomposing signals or vectors into simpler components within a vector space. Unlike bases, which require vectors to be linearly independent, frames allow redundancy, meaning they can contain more vectors than necessary. This redundancy provides greater freedom in representing signals more efficiently, particularly in the presence of noise or incomplete information \cite{casazza2000art}. These properties motivate the interpretation of quantum measurement results in terms of frames \cite{eldar2002optimal}, with the Naimark theorem \cite{beneduci2018notes} playing a central role in this context.\\

One area of interest is the study of a specific type of frame known as tight frames \cite{benedetto2002finite}, which are closely related to Positive Operator-Valued Measures (POVMs) \cite{yard2024introduction}. Tight frames have various applications in quantum information theory, including quantum tomography \cite{renes2004symmetric}, quantum state discrimination \cite{bergou2010discrimination}, state estimation \cite{derka1998universal}, Bell inequalities \cite{vertesi2010two}, quantum protocols \cite{gomez2016device}, quantum key distribution \cite{bennett1992quantum}, quantum sources \cite{jozsa2003entanglement}, and more.\\

As the goal is to design these measurements for quantum computers, a particular type of tight frame related to group theory via representation theory \cite{steinberg2016representation} has been selected for investigation. These frames, known as group frames \cite{waldron2018group}, are particularly interesting due to their connection with frame symmetries, which are expected to simplify the design process in quantum computing \cite{vale2018symmetry}. Among the many groups with various sizes and properties, abelian groups stand out due to their simple structure \cite{fuchs2015abelian}. This leads to the introduction of harmonic frames \cite{marshall2018number}, which is the primary focus of this thesis.\\

This work is structured as follows. Chapter 2 addresses the preliminaries necessary for a proper understanding of the research contents, in addition to introducing some examples to illustrate these concepts. First, the finite frame theory is studied (section \ref{frames}), with a primary focus on tight frames and their most important properties. For further reading on this topic, \cite{casazza2008finite} offers valuable insights.\\

Next, a review of representation theory is provided (section \ref{representation}), focusing on the structures and equivalences within this theory. An excellent resource on this subject is the book \cite{tung1985group}. With the knowledge acquired from these sections, we move on to the study of group frames (section \ref{group frames}), examining their relationship with the symmetry group of the frame and the structure of those group frames that are tight. A deeper exploration of this topic can be found in \cite{waldron2017tightframes}.\\

Following this, the definition of harmonic frames is reviewed (section \ref{harmonic}), along with some observations on Abelian groups and their representations, in order to characterize them in a way that is useful for the objectives of this thesis. For details on the structure of Abelian groups, \cite{piro2013fundamental} provides a concise overview, while \cite{chien2011unitary} delves deeply into harmonic frames.\\

\newpage

Next, we turn to concepts more related to information theory. In section \ref{state}, the definitions of state and separability are analyzed, along with their basic properties. A comprehensive discussion of these topics can be found in \cite{nielsen2010quantum}, and for additional perspectives on separability, \cite{horodecki2024multipartite} is a key reference.\\

Finally, section \ref{measures} introduces the definition of POVM and Naimark's theorem. For further exploration of these topics, \cite{paris2004modern} serves as an authoritative guide.\\

Chapter 3 focuses on the main results of this thesis. First, the separability of harmonic frames in bipartite states is analyzed (section \ref{bipartite}), with the primary findings summarized in propositions \ref{prop 3.1.1} and \ref{prop 3.1.2}. Subsequently, the study is extended to multipartite systems in section \ref{multipartite}, where the key result is presented as proposition \ref{prop 3.2.1}.\\

Next, harmonic frames associated with cyclic groups ($C_N$-frames) are explored in section \ref{Cn frames}. A design for implementing these frames in quantum computers is developed, leveraging Naimark's theorem (proposition \ref{prop 3.3.1}). Finally, the results are applied to simple examples on quantum computers, as detailed in sections \ref{example povm} and \ref{example separability}.\\

Finally, Chapter 4 presents the conclusions. A summary of the results obtained is provided in section \ref{summary}. This is followed by a reflection on these results in section \ref{reflections}, concluding with a discussion in section \ref{future research} on potential research directions that could further explore the contents of this thesis.\\

As will be analyzed throughout this work, this thesis lies at the intersection of group theory and quantum information theory, approached through the lens of frame theory. The results presented provide a fresh perspective on the design of quantum measurements, particularly in the context of quantum computing. By bridging these fields, this research not only contributes to the theoretical understanding of harmonic frames and their properties but also offers practical insights that may inspire new methods for implementing quantum measurements. This work aspires to serve as a foundation for further studies, opening avenues for innovative developments in quantum information science.

\chapter{Preliminaries}

\section{Finite frames}\label{frames}

Throughout this thesis we will be working with a $d$-dimensional Hilbert space $\mathcal{H}$ over a field $\mathbb{K}$, which will generally be $\mathbb{R}$ or $\mathbb{C}$. The elements in $\mathcal{H}$ will be denoted as $\ket{\psi}$ and the elements of the dual space $\mathcal{H}^{*}$ will be written as $\bra{\psi}$. In addition, the inner and outer product of $\ket{\psi},\ket{\phi}\in \mathcal{H}$ is described as $\braket{\psi}{\phi}$ and $\ket{\psi}\bra{\phi}$ respectively. We will also denote the norm $\sqrt{\braket{\psi}{\psi}}$ as $\vert\vert \psi\vert\vert$  and call $J$ a finite index set. \textbf{It is important to clarify that in this work \boldsymbol{\ket{\psi}} does not necessarily have norm one}. With this in mind, a finite frame can be defined as follows:

\begin{defi}
    
We say that the sequence $ (\ket{f_j})_{j\in J} \subset \mathcal{H} $ is a \textbf{frame} if for any $\ket{f}\in \mathcal{H}$ exists $A,B>0$ such that:
\begin{equation}\label{2.1}
A\,||f||^2\leq \sum_{j\in J} |\braket{f_j}{f}|^2\leq B\,||f||^2.
\end{equation}

\noindent When $A=B$, it is said that the set forms a \textbf{tight frame}. 

\end{defi}

In this section we will study basic properties of these mathematical objects, to later make the connection with states and POVMs, which are the main interest of this thesis. The following are some examples that apply to this definition:

\newpage

\begin{ex} \textbf{An orthonormal basis} forms a tight frame with $A=B=1$ due to the Parseval's identity.
    
\end{ex}

\begin{ex}\label{5} \textbf{Three equiangular vectors} in $\mathbb{C}^2$:
\begin{equation}
\ket{f_1}=\dfrac{1}{\sqrt{2}}\begin{bmatrix}
    1 \\
    1
\end{bmatrix}\,\,\,\,\, \ket{f_2}=\dfrac{1}{\sqrt{2}}\begin{bmatrix}
    1 \\
    \omega
\end{bmatrix} \,\,\,\, \, \ket{f_3}=\dfrac{1}{\sqrt{2}}\begin{bmatrix}
    1 \\
    \omega^2
\end{bmatrix}, 
\end{equation}

   \noindent where $\omega=e^{2\pi i/3}$ form a tight frame with $A=B=3/2$. This can be seen by calculating explicitly the sum of (\ref{2.1}) in an arbitrary vector and using that the sums of all unit roots is $0$.
\end{ex}

\begin{ex} The following vectors in $\mathbb{R}^2$ form a \textbf{frame that is not tight}:
\begin{equation}\ket{f_1}=\dfrac{1}{\sqrt{3}}\begin{bmatrix}
    1 \\
    2
\end{bmatrix}\,\,\,\,\, \ket{f_2}=\dfrac{1}{\sqrt{3}}\begin{bmatrix}
    2 \\
    1
\end{bmatrix},\end{equation}

    \noindent where $A=1/3$ and $B=3$. The above is obtained by calculating explicitly the sum of (\ref{2.1}) in an arbitrary vector and using the inequalities:
\begin{equation}
-4(a^2+b^2)\leq 8ab\leq 4(a^2+b^2),\end{equation}

\noindent where $a,b$ are the coordinates of the vector.
    
\end{ex}

The following operators can be useful when working with frames:

\begin{defi}

Let a sequence $(\ket{f_j})_{j\in J} \subset \mathcal{H}$. We define its \textbf{synthesis operator} as:

\begin{equation}\left\{ \begin{array}{lcc}
             V:l_2(J)\to \mathcal{H}\\
             V(a)=\displaystyle \sum_{j\in J} a_j\ket{f_j}

             \end{array} \right. .\end{equation}
\noindent Its dual $V^*$ is called the \textbf{analysis operator}, and the product $S:=VV^*$ is known as the \textbf{frame operator}.
    
\end{defi}

\newpage

For example, we can use these operators to provide a proof that the frames correspond to generating sets. For this purpose, let us note the following:
\begin{obs}\label{o2.1.1}
\begin{equation}\label{2.6}    
S\ket{f}=V(V^*\ket{f})=V(\braket{f_j}{f})_{j\in J}=\left(\sum_{j\in J}\ket{f_j}\bra{f_j} \right)\ket{f}.
\end{equation}
\end{obs}

\begin{prop}\label{4}
    A sequence $(\ket{f_j})_{j\in J}\subset \mathcal{H}$ is a frame if and only if its elements generate $\mathcal{H}$.\end{prop}

\begin{proof}  Suppose that the frame does not generate $\mathcal{H}$, then $V$ is not surjective, and therefore $S$ is not injective. Let $\ket{f}$ be a non-zero element in the kernel $Ker(S)$, then we have for (\ref{2.1}) and (\ref{2.6}) that $0<A\, ||f||^2\leq \bra{f}S\ket{f}=0$, which is a contradiction. On the other hand, if the set generates $\mathcal{H}$, then $V$ is surjective, and therefore $S$ is bijective. Consequently, we can take $A$ and $B$ as the smallest and largest eigenvalues of $S$ respectively. Since $S$ is positive definite, $A,B>0$, and by the spectral theorem, the bound $ A\vert\vert f\vert\vert^2\leq\bra{f}S\ket{f}\leq B\vert\vert f\vert\vert^2$ is obtained.

\end{proof}

The frame operator also allows for the characterization of tight frames, as shown below.

\begin{prop}\label{3} A sequence $(\ket{f_j})_{j\in J}$ in $\mathcal{H}$ is tight frame if and only if $S=AI_{\mathcal{H}}$, with $A>0$ and $I_{\mathcal{H}}$ the identity operator in $\mathcal{H}$. Additionally, a tight frame also satisfies that:
\begin{equation}
\sum_{j\in J} ||f_j||^2=dA \,\,\,\,\,\,\,\,\,\,\,\, \sum_{j\in J}\sum_{k\in J} |\braket{f_j}{f_k}|^2=dA^2.
\end{equation}
    
\end{prop}

\begin{proof}
   If $S=AI_{\mathcal{H}}$ with $A>0$, then $\bra{f}S\ket{f}=A\vert\vert f\vert\vert^2$, so the sequence is a tight frame. On the other hand, if the sequence form a tight frame, by the spectral theorem and definition (\ref{2.1}) applied to the eigenvectors of $S$, it follows that the only eigenvalue of $S$ is $A$, which implies that $S=AI_{\mathcal{H}}$. 
   
   \newpage
   
   Finally, by the cyclic property of the trace:
\begin{equation}dA=Tr(AI_{\mathcal{H}})=Tr(S)=Tr(V^*V)=\sum_{j\in J}||f_j||^2,\end{equation}
\noindent and
\begin{equation}dA^2=Tr(A^2I_{\mathcal{H}})=Tr(S^2)=Tr((V^*V)^2)=\sum_{j\in J}\sum_{k\in J}|\braket{f_j}{f_k}|^2.\end{equation}

\end{proof}

Another operator that is useful when analyzing frames is the \textbf{Gramian} $G:=V^*V$. This can be used to characterize the tight frame. To do so, let us note the following:

\begin{obs}

Given a tight frame $(\ket{f_j})_{j\in J}$ in $\mathcal{H}$ with constant $A$, if each vector is divided by $\sqrt{A}$, then we get again a tight frame with constant $1$ due to (\ref{2.6}) and the proposition \ref{3}. This is called a \textbf{normalized tight frame}.

\end{obs}

\begin{prop}\label{xdddd}
A frame $(\ket{f_j})_{j\in J}$ in $\mathcal{H}$ is normalized tight frame if and only if $G$ is an orthogonal projection.    
\end{prop}

\begin{proof}
    From the proposition \ref{3}, it follows that if the vectors form a normalized tight frame, then $G^2=V^*\underbrace{(VV^*)}_{=I}V=G=G^*$. On the other hand, if $G$ is a orthogonal projection, since the vectors generate by proposition \ref{4}, then $S$ is bijective, and thus:
\begin{equation}VG=SV=VG^2=S^2V,\end{equation}

\noindent thus, $V=SV$ and since $V$ is surjective, it follows that $S=I_{\mathcal{H}}$.
\end{proof}

The Gramian also allows us to encode the information of certain equivalences between frames.

\begin{defi} The sequences $(\ket{f_j})_{j\in J}$ and $(\ket{g_j})_{j\in J}$ are \textbf{unitarily equivalent} if there exists a unitary operator $U$ such that $U\ket{f_j}=\ket{g_j}$ for any $j\in J$. On the other hand, the sets are said to be \textbf{unitarily equivalent under rearrangement} if there exists a permutation $\sigma \in S_J$ such that $(\ket{f_j})_{j\in J}$ and $(\ket{g_{\sigma(j)}})_{j\in J}$ are unitarily equivalent. 
\end{defi}

\begin{ex}\label{6} The frame given by:
\begin{equation}\ket{g_1}=\dfrac{1}{\sqrt{2}}\begin{bmatrix}
    -1 \\
    1
\end{bmatrix}\,\,\,\,\, \ket{g_2}=\dfrac{1}{\sqrt{2}}\begin{bmatrix}
    -\omega^2 \\
    1
\end{bmatrix} \,\,\,\, \, \ket{g_3}=\dfrac{1}{\sqrt{2}}\begin{bmatrix}
    -\omega \\
    1
\end{bmatrix},\end{equation}

\noindent is unitarily equivalent under rearrangement to the frame from the example \ref{5}, with $\sigma=(2\, 3)$ (cycle notation) and $U=R_{\pi/2}$ the rotation matrix in $\pi/2$ (counterclockwise).
\end{ex}
\begin{prop}\label{hhh} Two frames are unitarily equivalent under rearrangement if their Gramian matrices are equal up to a permutation of rows and columns.
\end{prop}
\begin{proof}
  If $(\ket{f_j})_{j\in J}$ and $(\ket{g_{\sigma(j)}})_{j\in J}$ are frames unitarily equivalent, then $\braket{f_i}{f_j}=\braket{g_{\sigma(i)}}{g_{\sigma(j)}}$ for all $i,j\in J$, thus, their Gramian are equal, and as the Gramian of $(\ket{g_{j}})_{j\in J}$
  is equal to the Gramian of $(\ket{g_{\sigma(j)}})_{j\in J}$ with rows and columns permuted by $\sigma^{-1}$, it is concluded by transitivity. On the other hand, if the Gramian are equal up to a permutation of rows and columns $\sigma$, since the frames generate, we can take wlog $J=\{1\dots N\}$ and $(\ket{f_j})_{j=1}^d$ basis, with which we define $U\ket{f_j}=\ket{g_{\sigma(j)}}$ for $j\leq d$, then $U$ is unitary, and for $j>d$ and $i\leq d$ we have that $(U\ket{f_j}-\ket{g_{\sigma(j)}})^{*}\ket{g_{\sigma(i)}}=0$, therefore $U\ket{f_j}=\ket{g_{\sigma(j)}}$ for any $j\in J$.

\end{proof}
\begin{obs}
    The wlog is due to the fact that there exists a permutation $\tau$ such that $(\ket{f_{\tau(j)}})_{j=1}^d$ are basis, then $(\ket{f_{\tau(j)}})_{j\in J}$ and $(\ket{g_{\sigma\tau(j)}})_{j\in J}$ are unitarily equivalent, which implies the desired result.
\end{obs}
\begin{ex}\label{7} From the example \ref{6}, we can observe that, the Gramians $G_f$ and $G_g$ of $(\ket{f_j})_{j}$ and $(\ket{g_j})_{j}$ (respectively), are:
\begin{equation}G_f=\dfrac{1}{2}\begin{pmatrix}
   2&1+\omega&1+\overline{\omega} \\
      1+\overline{\omega}&2&1+\omega \\
      1+\omega &1+\overline{\omega}& 2

\end{pmatrix}\,\,\,\,\,\,\,\,\, G_g=\dfrac{1}{2}\begin{pmatrix}
   2&1+\overline{\omega}&1+\omega \\
      1+\omega&2&1+\overline{\omega} \\
      1+\overline{\omega} &1+\omega& 2

\end{pmatrix},\end{equation}

\noindent where we note that $G_f$ is equal to $G_g$ by permuting the 2nd and 3rd columns and rows.
    
\end{ex}

With the help of the Gramian, we can deduce the following proposition that provides a relationship between orthonormal bases and tight frames.

\begin{prop}\label{8} A frame $(\ket{f_j})_{j\in J}$ in $\mathcal{H}$ is normalized tight frame if and only if it is unitarily equivalent to an orthogonal projection of an orthonormal basis for $l_2(J)$.

\end{prop}

\begin{proof}
    If $(\ket{f_j})_{j\in J}$ is unitarily equivalent to an orthogonal projection $P$ of an orthonormal basis for $l_2(J)$, then, calling $V$ and $V^{\prime}$ the synthesis operator of the frame and the orthonormal basis (respectively) and $U$ the unitary mapping of the frame to it, we have that $UV=PV^{\prime}$, thus, $G=V^{\prime *}PV^{\prime}$, therefore is an orthogonal projection and by proposition \ref{xdddd} the frame is normalized and tight. On the other hand, if $(\ket{f_j})_{j\in J}$ is a normalized tight frame, we know by proposition \ref{xdddd} that its Gramian $G$ is an orthogonal projection. Let $\ket{j}$ be the canonical basis in $l_2(J)$, then we have that:
\begin{equation} (G\ket{i})^{*}(G\ket{j})=\bra{i}G\ket{j}=\braket{f_i}{f_j},\end{equation}
\noindent thus, by proposition \ref{hhh}, $(G \ket{j})_{j\in J}$ is unitarily equivalent to $(\ket{f_j})_{j\in J}$.
\end{proof}

The previous proposition allows us to represent a normalized tight frame in $l_2(J)$ under unitary equivalence. Moreover, it is also possible to obtain, from this representation alone, one in $\mathbb{K}^d$.

\begin{prop}\label{9}
    Let $P\in \mathbb{K}^{n\times n}$ orthogonal projection of rank $d$, then the columns of $V=[f_1,\dots,f_n]\in \mathbb{K}^{d\times n}$ form a normalized tight frame in $\mathbb{K}^{d}$ with Gramian $P$ if and only if the rows of $V$ are an orthonormal basis of the row space $row(P)$.\end{prop}

 \begin{proof}
     If $(f_j)_{j\leq n}$ is a normalized tight frame with Gramian $P$, then $VV^*=I$, thus, the rows of $V$ are an orthonormal basis and:
\begin{equation}row(G)=row(V^*V)\subset row(V)=row(\underbrace{(VV^*)}_{=I}V)\subset row(V^*V)=row(G),\end{equation}

     \noindent and since $G=P$, it follow that $row(V)=row(G)=row(P)$. On the other hand if the rows of $V$ are an orthonormal basis then $VV^*=I$, thus, $(f_j)_{j\leq n}$ is a normalized tight frame and it follow that $G$ is a orthogonal projection with $row(G)=row(V)=row(P)$, then, the images $Im(G)=Im(P)$, therefore $G=P$.
     
 \end{proof}

 \begin{ex} Normalizing the frame from the example \ref{7}, we can observe that:
\begin{equation}G_f=\dfrac{1}{3} \begin{pmatrix}
   2&-\overline{\omega}&-\omega \\
      -\omega&2&-\overline{\omega}\\
      -\overline{\omega}&-\omega& 2
      \end{pmatrix}.\end{equation}
\noindent By the proposition \ref{8}, the columns of $G_f$ are a canonical copy of $(\sqrt{2/3}\, \ket{f_j})_{j\in J}$ in $\mathbb{C}^3$ up to unitarily equivalence and by proposition \ref{9}, we can recover the tight frame in $\mathbb{C}^2$ using Gram Schmidt in the rows of $G_f$:
\begin{equation}V=\sqrt{\dfrac{2}{3}} \begin{pmatrix}
   1&-\overline{\omega}/2&-\omega/2 \\
      0&\sqrt{3}/2 &-\sqrt{3}\,\overline{\omega}/2\\
      \end{pmatrix}.\end{equation}
    
 \end{ex}

 \newpage

 \section{Representation theory}\label{representation}
In this section we will work with a finite group $\mathcal{G}$ and we denote $Gl(\mathcal{H})$ as the group of invertible operators of $\mathcal{H}$ on $\mathcal{H}$. The main object of study will be given by the following definition.

\begin{defi}
    We call $D:\mathcal{G}\to Gl(\mathcal{H})$ a \textbf{representation} of $\mathcal{G}$ in $\mathcal{H}$ if $D$ is a group homomorphism. The dimensionality of the representation is given by the dimension of $\mathcal{H}$.

\end{defi}

This definition will be relevant later for the study of a specific type of frames, called group frames. With this objective in mind, it is important to note that
    a representation $D$ induces a linear action of $\mathcal{G}$ on 
$\mathcal{H}$ which is given by $g\ket{v}:=D(g)\ket{v}$. The following are some examples of representations:

\begin{ex} \textbf{The trivial representation} of $\mathcal{G}$ in $\mathcal{H}$ is given by the trivial homomorphism $D(g)=I_{\mathcal{H}}$.

\end{ex}

\begin{ex} \textbf{The determinant associated with a representation \boldsymbol{D}}, i.e $Det_D(g)=Det(D(g))$, is a one-dimensional representation.

\end{ex}

\begin{ex} \textbf{The one-dimensional representations of the cyclic group \boldsymbol{C_n}} are of the form $D(a^k)=\omega^{k}$, where $\omega=D(a)$ is an $n-th$ root of unity in $\mathbb{K}$.
\end{ex}

To study the representations of a group, equivalence under similarity relations is used, which allows for simplifying the analysis.

\begin{defi}
    Two representation $D_1$,$D_2$ of $\mathcal{G}$ in Hilbert spaces $\mathcal{H}_1$ and $\mathcal{H}_2$ respectively are \textbf{equivalent} if exists
    a invertible linear operator $A:\mathcal{H}_1\to \mathcal{H}_2$ such that for any $g\in \mathcal{G}$:
  \begin{equation}
        D_1(g)=A^{-1}D_2(g)A.
    \end{equation}

\end{defi}

\newpage

The following result allows us to reduce the study to a particular type of representations.

\begin{prop}\label{10}
    Every representation $D$ of $\mathcal{G}$ in $\mathcal{H}$ is equivalent to a unitary representation.
\end{prop}

\begin{proof}
    Let us consider the following operator:
\begin{equation}H:=\sum_{g\in \mathcal{G}} D(g)D^{*}(g).\end{equation}

\noindent Since $H$ is the sum of positive definite operators, it is positive definite; therefore, it is possible to define $\sqrt{H}$ which is positive definite. Finally, by the rearrangement lemma, we note that the operator $\Tilde{D}(g):=\sqrt{H}^{-1}D(g)\sqrt{H}$ is unitary. In fact:
$$\Tilde{D}(g)\tilde{D}^{*}(g)=\left(\sqrt{H}^{-1}D(g)\sqrt{H}\right)\left(\sqrt{H} D^{*}(g)\sqrt{H}^{-1}\right))=\sqrt{H}^{-1}\left(\sum_{h\in \mathcal{G}} D(gh)D^{*}(hg)\right)\sqrt{H}^{-1}$$
\begin{equation}=\sqrt{H}^{-1}H\sqrt{H}^{-1}=I_{\mathcal{H}}.\end{equation}

\end{proof}

From this, it is possible to study the basic constituents of representations.

\begin{defi} A subspace $W$ of 
$\mathcal{H}$ is said to be \textbf{invariant} with respect to the representation 
$D$ if for every $g\in \mathcal{G}$, the image $D(g)W\subset W$. If the representation has exactly two invariant subspaces (the trivial $\{0\}$ and $\mathcal{H}$), it is said to be \textbf{irreducible}; otherwise, it is called \textbf{reducible}. In the case that 
$D$ is reducible and the orthogonal complement $W^{\perp}$ is invariant with respect to $D$, we say that $D$ is $\textbf{completely reducible}$.

\end{defi}

\begin{obs}\label{11}
    If $D$ is completely reducible, writing space as a direct sum $\mathcal{H}=W\bigoplus W^{\perp}$, $D$ can be viewed as a direct sum of operators $D=D_1\oplus D_2$, where $D_1(g)=D(g)\vert_{W}$ and  $D_2(g)=D(g)\vert_{W^{\perp}}$ are representations in $W$ and $W^{\perp}$ respectively. 
\end{obs}

\newpage
\begin{prop}\label{12}
    If $D$ is a reducible unitary representation, then it is completely reducible.
\end{prop}

\begin{proof}
    Let $W$ non-trivial invariant subspace, $\ket{w^{\perp}}\in W^{\perp}$ and $\ket{w}\in W$, then:
\begin{equation}\bra{w}D(g)\ket{w^{\perp}}=\bra{w^{\perp}}D^{*}(g)\ket{w}=\bra{w^{\perp}}D^{-1}(g)\ket{w}=\bra{w^{\perp}}\underbrace{D(g^{-1})\ket{w}}_{\in W}=0,\end{equation}
    
\noindent therefore, $W^{\perp}$ is invariant, i.e, $D$ is completely reducible.
\end{proof}

With this result, it is possible to reduce the study of representations to those that are irreducible. There are two very important lemmas when working with irreducible representations, which are presented below:

\begin{lem}\textbf{(Schur 1)}
    If $D$ is an irreducible representation and $A\in L(\mathcal{H})$ such that the commutator $[A,D(g)]=0$ for any $g\in \mathcal{G}$, then $A=\lambda I_{\mathcal{H}}$ with $\lambda\in \mathbb{K}$.
\end{lem}

\begin{proof}
    Note that $Ker(A)$ is invariant with respect to $D$, since for $\ket{v}\in Ker(A)$:
\begin{equation}AD(g)\ket{v}=D(g)A\ket{v}=0.\end{equation}

\noindent But $D$ is irreducible, so $Ker(A)\in\{\{0\},\mathcal{H}\}$, if $Ker(A)=\mathcal{H}$, $A=0$, and if $Ker(A)=\{0\}$, then $A$ has a non-zero eigenvalue $\lambda\in \mathbb{K}$. Since $A-\lambda I_{\mathcal{H}}$ also commutes with $D(g)$ for any $g\in \mathcal{G}$ and $Ker(A-\lambda I_{\mathcal{H}})\neq \{0\}$, it follows that $A-\lambda I_{\mathcal{H}}=0$, that is $A=\lambda I_{\mathcal{H}}$.

\end{proof}

\begin{lem}\textbf{(Schur 2)} Let $D_1$, $D_2$ irreducible representations in $\mathcal{H}_1$ and $\mathcal{H}_2$ respectively. If $A\in L(\mathcal{H}_2,\mathcal{H}_1)$ is such that $D_1(g)A=AD_2(g)$ for any $g\in \mathcal{G}$, then $A=0$ or there is an isomorphism $\mathcal{H}_1\cong \mathcal{H}_2$ and $D_1$, $D_2$ are equivalent.

\end{lem}
    
\begin{proof} Note that $Im(A)$ is invariant with respect to $D_1$, since for $\ket{v}\in Im(A)$, exists $\ket{w}\in \mathcal{H}_2$ such that $\ket{v}=A\ket{w}$, then:

\begin{equation}D_1(g) \ket{v}=D_1(g)A\ket{w}=AD_2(g)\ket{w}\in Im(A),\end{equation}

\noindent thus, $Im(A)\in \{0,\mathcal{H}_1\}$. If $Im(A)=\{0\}$, then $A=0$. On the other hand, if $Im(A)=\mathcal{H}_1$, $A$ is surjective and therefore bijective, hence $\mathcal{H}_1\cong \mathcal{H}_2$ and $D_2=A^{-1} D_1 A$.     
\end{proof}

The first Schur lemma allows us to obtain information about the irreducible representations of Abelian groups.

\begin{cor}\label{x}
  The irreducible representations of an Abelian group are one-dimensional.  
\end{cor}

\begin{proof} Let $D$ a irreducible representation of $\mathcal{G}$ abelian group, then $[D(g_1),D(g_2)]=0$ for any $g_1,g_2\in \mathcal{G}$, thus, by the first Schur lemma, $D(g)=\lambda_{g} I_\mathcal{H}$, but since $D$ is irreducible, this is only possible if $Dim(\mathcal{H})=1$.
    
\end{proof}

\begin{obs}
    All one-dimensional representations are irreducible, since the subspaces of $\mathcal{H}\cong \mathbb{K}$ are exactly the ideals, which characterize a field as having trivial.
\end{obs}

Another important result of Schur's lemmas are the orthogonality relations.

\begin{cor}\label{13}
        Let $\{D^i\}_{i\in I}$ be the non-equivalent irreducible representations of $\mathcal{G}$ in $\mathcal{H}_i$. For any $B\in L(\mathcal{H}_j,\mathcal{H}_i)$ exists $\lambda_{B,i}$ such that:

    \begin{equation}A:=\sum_{g\in \mathcal{G}}D^{i}(g^{-1})BD^{j}(g)=\lambda_{B,i}\, \delta_{i,j}\,I_{\mathcal{H}_i}.\end{equation}

\end{cor}

\begin{proof} We can observe that, by the rearrangement lemma:
\begin{equation}AD^{j}(h)=\sum_{g\in \mathcal{G}}D^{i}(g^{-1})BD^{j}(gh)=D^{i}(h)\sum_{g\in \mathcal{G}}D^{i}((gh)^{-1})BD^{j}(gh)=D^{i}(h) A,\end{equation}
\noindent where $i\neq j$, by the second Schur lemma, $A=0$, and when $i=j$, by the first lemma exists $\lambda_{B,i}\in \mathbb{K}$ such that $A=\lambda_{B,i}I_{\mathcal{H}_i}$, thus, the statement is concluded.
\end{proof}

\newpage

\begin{obs} Considering some basis for $\mathcal{H}_i$ and $\mathcal{H}_j$, we can choose $B_{k,l}=\delta_{k,r}\, \delta_{l,s}$, then:
\begin{equation}\sum_{g\in \mathcal{G}} D_{k,r}^{i}(g^{-1})D_{s,l}^{j}(g)=\lambda_{r,s,i}\,\delta_{i,j}\,\delta_{k,l},\end{equation}
\noindent in this case $i=j$ and $k=l$:
\begin{equation}\sum_{g\in \mathcal{G}} D_{s,k}^{i}(g)  D_{k,r}^{i}(g^{-1})=\lambda_{r,s,i}.\end{equation}
\noindent Summing over $k$ we obtain that $\lambda_{r,s,i}=\delta_{r,s}(n_{\mathcal{G}}/n_i)$ where $n_\mathcal{G}$ is the order of $\mathcal{G}$ and $n_i$ the dimension of $D^i$. Therefore:
\begin{equation}\sum_{g\in \mathcal{G}} D_{k,r}^{i}(g^{-1})D_{s,l}^{j}(g)=\dfrac{n_\mathcal{G}}{n_i}\,\delta_{i,j}\,\delta_{k,l}\,\delta_{r,s}.\end{equation}
\noindent When the representation is chosen to be unitary and the basis is orthonormal, it follows that $D^{i}_{k,r}(g^{-1})=\overline{D^{i}_{r,k}(g)}$, therefore, they are called orthogonality relations.

\end{obs}

When studying the representations of a group, it is useful to define the quantity:
\begin{equation}\xi^{i}(g):=Tr(D^{i}(g)),\end{equation}
\noindent which is called the \textbf{character}.

\begin{obs}\label{xd}
    
Choosing $r=k$ and $s=l$, and summing over $k$ and 
$l$ in the orthogonality relations, we obtain that:
\begin{equation}\sum_{g\in \mathcal{G}} \xi^i(g^{-1})\xi^{j}(g)=n_\mathcal{G} \, \delta_{i,j}.\end{equation}
\noindent Furthermore, since the trace is invariant under similarity transformation, it follows that $\xi^{i}(g^{-1})=\overline{\xi^{i}(g)}$.

\end{obs}

\newpage

\section{Group Frames}\label{group frames}
In this section we will denote $U(\mathcal{H})$ the set of unitary operators of $\mathcal{H}$ in $\mathcal{H}$ and $U(n)$ the unitary matrix of $n\times n$. We will now study the theory of the following types of frames:

\begin{defi}
    A frame $(\ket{f_{g}})_{g\in \mathcal{G}}$ in $\mathcal{H}$ indexed by $\mathcal{G}$ is called a \boldsymbol{\mathcal{G}}\textbf{-frame} if exists a unitary representation $D:\mathcal{G}\to U(\mathcal{H})$ such that:
\begin{equation}g\ket{f_h}:=\ket{f_{gh}}.\end{equation}

\end{defi}

\begin{obs}
    
Intuitively, $\mathcal{G}$-frames are those frames formed from the orbit of a unitary linear action of the group $\mathcal{G}$ on some vector $\ket{f_e}$. If the definition holds for a non-unitary action, then the frame will be similar to a $\mathcal{G}$-frame by proposition \ref{10}. The converse is also true.

\end{obs}

Below are some examples and non-examples of group frames:

\begin{ex}
\textbf{A basis} $(\ket{f_j})_{j\in J}$ in $\mathcal{H}$ can be indexed by the elements of any group $\mathcal{G}$ of order $\vert J\vert$. With this, it is possible to define $D:\mathcal{G}\to Gl(\mathcal{H})$ such that $D(g)\ket{f_{h}}:=\ket{f_{gh}}$, therefore the vectors are similar to a $\mathcal{G}$-frame.

\end{ex}

\begin{ex} \textbf{The three equiangular vectors in \boldsymbol{\mathbb{C}^2}} (example \ref{5}) can be indexed by the elements of the cyclic group $C_3=\langle a\rangle$, where $\ket{f_{a^{j-1}}}:=\ket{f_j}$. Considering the homomorphism $D:C_3\to U(2)$:

\begin{equation}D(a):=\begin{pmatrix}
    1 & 0\\
    0& w
\end{pmatrix},\end{equation}

\noindent it can be noted that the vectors form a $C_3$-frame.

\end{ex}

\begin{ex} The vectors $\{\ket{1},\ket{2},2\,\ket{1}\}$ in $\mathbb{R}^2$ form a \textbf{frame not similar to any \boldsymbol{\mathcal{G}} frame}. Indeed, if it were so, it would be a $C_3$-frame, and thus there would exist a representation with an element $g\in \mathcal{G}$ that satisfies $D(g)\ket{1}=\ket{2}$, then, $D(g)2\ket{1}=2\,\ket{2}$, which is a contradiction, since $2\,\ket{2}$ does not belong to the frame.

\end{ex}

\newpage

To determine if a frame is similar to a $\mathcal{G}$-frame, it is useful to define the symmetry group of the frame:

\begin{defi}
    Let $\Phi:=(\ket{f_j})_{j\in J}$ be a frame in $\mathcal{H}$, we define its \textbf{symmetry group} through the following set with the composition operation:
\begin{equation}Sym(\Phi):=\{\sigma\in S_{J} : \exists L_{\sigma}\in Gl(\mathcal{H}): \,\forall j\in J,\, L_{\sigma}\ket{f_j}=\ket{f_{\sigma(j)}}\}.\end{equation}
    
\end{defi}

\begin{prop}   Let $\Phi:=(\ket{f_j})_{j\in J}$ be a frame in $\mathcal{H}$, This is similar to a $\mathcal{G}$-frame for some group $\mathcal{G}$ if and only if there exists a transitive subgroup $K\leq Sym(\Phi)$ of order $\vert J\vert$.
    
\end{prop}

\begin{proof}
    If exists $K\leq Sym(\Phi)$ transitive of order $\vert J\vert$, then there is a $j\in J$ such that $Kj=J$, thus, $\Phi=(\ket{f_{\sigma(j)}})_{\sigma \in K}$. Reindexing $\ket{f_{\sigma(j)}}=\ket{f_{\sigma}}$ and defining $D:K\to Gl(\mathcal{H})$ as $D(\sigma):=L_{\sigma(j)}$ it is shown that $\Phi$ is similar to a K-frame. On the other hand, if $\Phi$ is similar to a $\mathcal{G}$-frame, then exists a bijection $\phi:J\to G$ such that $\ket{f_j}=\ket{f_{\phi(j)}}$, which allows defining  $\sigma_{g}(j)=\phi^{-1}(g\phi(j))\in S_J$. We can note that the image $\sigma_{G}\cong G$ from $\sigma$, moreover, $\sigma_{G}\phi^{-1}(e)=J$, and defining $L_{\sigma_g}:=D(g)$ it is shown that $\sigma_{G}$ is a transitive subgroup of $Sym(\Phi)$ with order $\vert J\vert$.

\end{proof}

To calculate the elements of $Sym(\Phi)$ we can make use of the permutation operators:
\begin{equation}\left\{ \begin{array}{lcc}
             P_{\sigma}:l_2(J)\to l_2(J)\\
             P_{\sigma}\ket{j}=\ket{\sigma(j)}

             \end{array} \right. .\end{equation}

\begin{obs}\label{c}
    Note that $\sigma\in Sym(\Phi)$ if and only if exists $L_{\sigma}\in Gl(\mathcal{H})$ such that for any $j\in J$:
\begin{equation}L_{\sigma}V\ket{j}=L_{\sigma}\ket{f_j}=\ket{f_{\sigma(j)}}=VP_{\sigma}\ket{j},\end{equation}
\noindent that is, $L_{\sigma}V=VP_{\sigma}$.

\end{obs}

\newpage

\begin{obs}
    Recall that if $\Phi$ is a frame, then 
$S$ is positive definite, so it is possible to define \textbf{the canonical tight frame associated with \boldsymbol{\Phi}}:
\begin{equation}\Phi_{c}:=(S^{-1/2}\ket{f_j})_{j\in J},\end{equation}
\noindent which is a normalized tight frame, since if $V_c$ is its synthesis operator:
\begin{equation}S_c=V_c V_c^{*}=S^{-1/2}VV^{*}S^{-1/2}=S^{-1/2}SS^{-1/2}=I_{\mathcal{H}}.\end{equation}

\end{obs}

\begin{obs}\label{a}
    Note that if $\Phi=(\ket{f_j})_{j\in J}$ and $\Psi$ are similar frames, then exists $Q\in Gl(\mathcal{H})$ such that $\Psi=(Q\ket{f_j})_{j\in J}$, thus, if $\sigma\in Sym(\Phi)$, defining: 
    \begin{equation}L^{\Psi}_{\sigma}:=QL^{\Phi}_{\sigma}Q^{-1}\in Gl(\mathcal{H}),\end{equation}
  \noindent then, for any $j\in J$, $L^{\Psi}_{\sigma}Q\ket{f_j}=Q\ket{f_{\sigma(j)}}$, that is $\sigma\in Sym(\Psi)$, theferore $Sym(\Phi)\subset Sym(\Psi)$. Interchanging $\Phi$ for $\Psi$ and repeating the argument, we can prove that $Sym(\Phi)=Sym(\Psi)$.

\end{obs}

\begin{obs}\label{b}
If $\Phi=(\ket{f_j})_{j\in J}$ is a tight frame with constant $A$, then, for $\sigma\in Sym(\Phi)$, it follows by observation \ref{c} that:
\begin{equation}L_{\sigma}L_{\sigma}^{*}=\dfrac{1}{A}\, L_{\sigma}VV^{*}L_{\sigma}^{*}=\dfrac{1}{A}\, VP_{\sigma}P_{\sigma}^{*}V^{*}=\dfrac{1}{A}\, VV^{*}=I_{\mathcal{H}}.\end{equation}

\end{obs}

\begin{prop} Let $\Phi$ a frame in $\mathcal{H}$, then $\sigma\in Sym(\Phi)$ if and only if
$P_{\sigma}^{*} G_{\Phi_c} P_{\sigma}=G_{\Phi_c}$, where $G_{\Phi_c}$ is the Gramian of the canonical tight frame associated.

\end{prop}

\begin{proof}
    If $\sigma\in Sym(\Phi)$, then by the observations \ref{a} and \ref{b}, $\sigma\in Sym(\Phi_c)$ with $L^{c}_{\sigma}\in U(\mathcal{H})$, thus, by the observation \ref{c}:
\begin{equation}P_{\sigma}^{*}G_{\Phi_c} P_{\sigma}=(V_c P_{\sigma})^{*}(V_c P_{\sigma})=(L^{c}_{\sigma}V)^{*}(L^{c}_{\sigma}V)=V_c^{*} L^{c*}_{\sigma}L^{c}_{\sigma}V_c=G_{\Phi_c}.\end{equation}

\newpage

\noindent On the other hand, if $P_{\sigma}^{*}G_{\Phi_c} P_{\sigma}=G_{\Phi_c}$, then we define:
\begin{equation}L^{c}_{\sigma}:=V_{c}P_{\sigma}V_{c}^{*},\end{equation}
\noindent thus:
\begin{equation}L^{c*}_{\sigma}L^{c}_{\sigma}=V_c(P_{\sigma}^{*}V_{c}^{*}V_cP_{\sigma})V_{c}^{*}=V_cV_c^{*}V_cV_{c}^*=I_{\mathcal{H}},\end{equation}
\noindent and $L_{\sigma}V_c=V_c P_{\sigma}$, therefore, $\sigma\in Sym(\Phi)$.

\end{proof}
    
From Observation \ref{11} and Proposition \ref{12}, a unitary representation induces a decomposition of the space into a direct sum of invariant spaces. In particular, a $\mathcal{G}$-frame also induces this decomposition.

\begin{obs}
    Let $(g\ket{v})_{g\in \mathcal{G}}$ a $\mathcal{G}$-frame in $\mathcal{H}$, and:
\begin{equation}\mathcal{H}=\bigoplus_{i\leq n} W_i,\end{equation}

\noindent the decomposition in orthogonal irreducible invariant spaces. If $\ket{v}=\oplus_{i\leq n} \ket{w_i}$, then:
\begin{equation}(g\ket{v})_{g\in \mathcal{G}}=\bigoplus_{i\leq n}(g\ket{w_i})_{g\in \mathcal{G}},\end{equation}
\noindent where $(g\ket{w_i})_{g\in \mathcal{G}}$ are \textbf{irreducible \boldsymbol{\mathcal{G}}-frames in \boldsymbol{W_i}}.

\end{obs}

\begin{prop}
    Let $(g\ket{v})_{g\in \mathcal{G}}$ a irreducible $\mathcal{G}$-frame in $\mathcal{H}$, then is tight.
\end{prop}

\begin{proof}
    By the corolary \ref{13}:
 \begin{equation}0<S=\sum_{g\in \mathcal{G}}g\ket{v}\bra{v}g^{*}=\underbrace{\lambda_v}_{>0}\, I_\mathcal{H},\end{equation}
\noindent concluding thus, for the proposition \ref{3}, that the $\mathcal{G}$-frame is tight.
\end{proof}

   \newpage

   \begin{obs}\label{14}
       
   In a $\mathcal{G}$-tight frame $(g\ket{v})_{g\in \mathcal{G}}$ in $\mathcal{H}$ with constant $\lambda$, it follows by the proposition \ref{3} that:
\begin{equation}n_\mathcal{G}\,\vert\vert v\vert\vert^2=\sum_{g\in \mathcal{G}} \vert \vert gv\vert\vert^2=d\,\lambda,\end{equation}
\noindent that is to say, $\lambda=\dfrac{n_\mathcal{G}}{d}\,\vert\vert v\vert\vert^2$.

   \end{obs}

With this, we can reduce the study to irreducible $\mathcal{G}$-frames. However, what conditions must the components of the direct sum satisfy for the $\mathcal{G}$-frame to be tight?

\begin{prop}\label{y}
    Let $(g\ket{v})_{g\in \mathcal{G}}$ a $\mathcal{G}$-frame in $\mathcal{H}$ and:

    \begin{equation}\mathcal{H}=\bigoplus_{i\leq n} W_i,\end{equation}

    \noindent the decomposition in orthogonal irreducible invariant spaces. If $\ket{v}=\oplus_{i\leq n}\ket{w_i}$, then the $\mathcal{G}$-frame is tight if and only if for $i\neq j$:
\begin{equation}\dfrac{\vert \vert w_i\vert\vert^2}{\vert \vert w_j\vert\vert^2}=\dfrac{n_i}{n_j}\,\,\,\,\,\,\,\, ;\,\,\,\,\,\,\,\, \sum_{g\in \mathcal{G}}g\ket{w_i}\bra{w_j}g^{*}=0.\end{equation}

\end{prop}

\begin{proof}
    Note that, by the observation \ref{14}, $(g\ket{v})_{g\in \mathcal{G}}$ is a $\mathcal{G}$-tight-frame if and only if:
$$\dfrac{n_\mathcal{G}}{d}\, \vert\vert v\vert\vert^2 \, I_{\mathcal{H}}=\sum_{g\in \mathcal{G}} g\ket{v}\bra{v}g^{*}=\bigoplus_{i\leq n} \sum_{g\in \mathcal{G}} g \ket{w_i}\bra{w_i}g^{*}+\bigoplus_{i,j:i\neq j}\sum_{g\in \mathcal{G}} g\ket{w_i}\bra{w_j}g^{*}$$
\begin{equation}=\bigoplus_{i\leq n}\dfrac{n_\mathcal{G}}{n_i}\vert\vert w_i\vert\vert^2\, I_{W_i}+\bigoplus_{i,j:i\neq j}\sum_{g\in \mathcal{G}} g\ket{w_i}\bra{w_j}g^{*}.\end{equation}
\noindent Evaluating in $\ket{f_j}\in W_j$ for $j\leq n$, it is concluded by the uniqueness of the expressions in the direct sum, that for $i\neq j$:
\begin{equation}\dfrac{\vert\vert v\vert\vert^2}{d}=\dfrac{\vert\vert w_j\vert\vert^2}{n_j}\,\,\,\,\,\,\,\, ;\,\,\,\,\,\,\,\, \sum_{g\in \mathcal{G}}g\ket{w_i}\bra{w_j}g^{*}=0.\end{equation}

\end{proof}

\newpage

\section{Harmonic Frames}\label{harmonic}
In the previous section we introduced the concept of group frame along with its main properties. Due to their symmetry, these are our main study candidates for implementation in quantum computers. Among all the groups, we highlight those that are abelian due to the simplicity of their structure. This motivates the following definition:
\begin{defi}
    
Let $\mathcal{G}$ be an abelian group. We say that $(\ket{f_j})_{j\in J}$ is a \textbf{Harmonic frame} if it is a $\mathcal{G}$-tight frame.
\end{defi}

\begin{obs}
    By the corollary \ref{x}, we know that the irreducible representations of an abelian group are one-dimensional, then $D^{i}(g)=\xi^{i}(g)$. Moreover, they satisfy that $(\xi^{i}(g))^{n_\mathcal{G}}=1$, that is, are $n_\mathcal{G}$ roots of unity, therefore, \textbf{we will study the representations in $\mathbb{K}=\mathbb{C}$}.
    
\end{obs}

\begin{obs}
    
    By the Structure Theorem for Abelian Groups \cite{dummit2004abstract}, any finite abelian group can be decomposed as a direct product of cyclic groups of prime power orders, i.e:
  \begin{equation}
    \mathcal{G}\cong C_{p_0^{t_0}}\times\dots\times C_{p_{m-1}^{t_{m-1}}},\end{equation}

\noindent thus, the representations of $\mathcal{G}$ can be viewed as representations in the product.

\end{obs}

\begin{obs}
     Note that given a representation in $D:\mathcal{G}_1\times \mathcal{G}_2\to Gl(\mathcal{H})$. If we define
     $D_1(g_1):=D(g_1,e_2)$ and $D_2(g_2):=D(e_1,g_2)$ with $e_1$ and $e_2$ the identities in $\mathcal{G}_1$ and $\mathcal{G}_2$ respectively, then:
\begin{equation}D(g_1,g_2)=D((g_1,1)(1,g_2))=D(g_1,1)D(1,g_2)=D_1(g_1)D_2(g_2). \end{equation}
\noindent In addition, $D_1,D_2$ are representation of $\mathcal{G}_1,\mathcal{G}_2$ in $\mathcal{H}$ respectively.
\end{obs}




\newpage

\begin{obs}\label{21}
    With the above observations of this section, we can deduce that if we call $\hat{\mathcal{G}}$ the irreductible characters of the abelian group $\mathcal{G}$, then $\mathcal{G}\cong \hat{\mathcal{G}}$ via the isomorphism:
\begin{equation}\left\{ \begin{array}{lcc}
             \Psi:C_{n_0}\times \dots\times C_{n_{m-1}}\to \overbrace{C_{n_0}\times \dots \times C_{n_{m-1}}}\\
            \Psi(a_0^{r_0},\dots,a_{m-1}^{r_{m-1}})(a_0^{s_0},\dots,a_{m-1}^{s_{m-1}})=\displaystyle \prod_{j=0}^{m-1}\Psi_j(a_j^{r_j})(a_j^{s_j})

             \end{array} \right. ,\end{equation}
\noindent where $n_j:=p_j^{t_j}$ and $\Psi_j$ are the isomorphism between $C_{n_j}$ and $\hat{C}_{n_j}$, given by:
\begin{equation}\left\{ \begin{array}{lcc}
             \Psi_j:C_{n_j}\to  \hat{C}_{n_j}\\
            \Psi_j(a_j^{r_j})(a_j^{s_j})=w_j^{r_j\,s_j}

             \end{array} \right. ,\end{equation}
             
\noindent and $w_j=e^{2\pi i/n_j}$.
\end{obs}

\begin{obs}\label{20}
    By the proposition \ref{y}, the Harmonic frames can be written as:
\begin{equation}\bigoplus_{i\leq n} (\xi^{i}(g)\ket{w_i})_{g\in \mathcal{G}},\end{equation}

\noindent such that for $i\neq j$, the representations are not equivalent and the norms of $\ket{w_i}$ and $\ket{w_j}$ are equals. Normalizing the vectors and mapping $\ket{w_j}\to \ket{j}$, \textbf{the harmonic frames can be viewed as a selection of \boldsymbol{d} rows from the character table of the group \boldsymbol{\mathcal{G}}}, which is defined as $(\xi^{i}(g_j))_{i,j\leq n_\mathcal{G}}$ (by convention, we take $g_1=e$ and $\xi^{1}$ the trivial representation), 
moreover, the isomorphism between $\mathcal{G}$ and $\hat{\mathcal{G}}$ also allows \textbf{harmonic frames to be seen as a selection of \boldsymbol{d} columns from the character table, such that each one is associated with a sequence \boldsymbol{J} of different elements of $\mathcal{G}$}.

\end{obs}

\newpage

\section{Quantum states and multipartite separability}\label{state}

Throughout this section we will denote $L(\mathcal{H})$ as the linear operators of $\mathcal{H}$ in $\mathcal{H}$. With this we can define the following concept of physical interest.

\begin{defi}
    We say that $\rho \in L(\mathcal{H})$ is a \textbf{quantum state} if $\rho=\rho^{*}\geq 0$ and $Tr(\rho)=1$. A quantum state is called \textbf{pure} if its only representation as a convex sum of states is the trivial one (i.e, if $\rho=p_1\rho_1+\dots+p_m \rho_m$ then $\rho_i=\rho$). Otherwise, we say it is \textbf{mixed}.
\end{defi}

\textbf{From now on we will only mention quantum states simply as states}. An example is shown below to provide an intuition for this definition.

\begin{ex}
 Consider an experiment whose classical outcomes are given by $J$. We define the associated Hilbert space as $\mathcal{H}:=\l^2(J)$. With this, the representation of these outcomes as pure states in $\mathcal{H}$ is given by:
\begin{equation}\left\{ \begin{array}{lcc}
             E_{i,i}:\mathcal{H}\to  \mathcal{H}\\
            E_{i,i}\ket{j}:=\delta_{i,j}  \,\ket{j}

             \end{array} \right. ,\end{equation}
\noindent with $i,j\in J$. Taking this into account, the probability state $p:=(p_j)_{j\in J}$ associated with the outcomes J, can be seen as:
\begin{equation}\rho_p:=\sum_{j\in J}p_j \, E_{j,j},\end{equation}
\noindent this states are mixed.
\end{ex}

To identify pure states in a simpler way, the following characterization can be used.

\begin{prop}\label{pp}
    A operator $\rho\in L(\mathcal{H})$ is a pure state if and only if $\rho^2=\rho^{*}=\rho$.
\end{prop}

    \begin{proof}
        If $\rho$ is a pure state, then for being a state $\rho=\rho^{*}$, thus, by spectral theorem exists a orthonormal basis of eigenvector $\ket{v_j}$ with eigenvalues $p_j$, thus:

        \begin{equation}\rho=\sum_{j=1}^{d} p_j \ket{v_j}\bra{v_j},\end{equation}

\noindent since $Tr(\rho)=1$, this is a convex sum of states, thus, for $i\neq j$, $\ket{v_i}\bra{v_i}=\ket{v_j}\bra{v_j}$, therefore, $\rho^2=\rho$. If $\rho^2=\rho^*=\rho$, by spectral theorem $\rho=\ket{v}\bra{v}$. Let a state convex sum representation $\rho=p_1\rho_1+\dots+p_m\rho_m$, if $\ket{w}\in Ker(\rho)$ then:

\begin{equation}\sum_{j\leq n} p_j \bra{w}\rho_i\ket{w}=0,\end{equation}

\noindent but by spectral decomposition in $\rho_i\geq 0$, this occurs only if $\ket{w}\in Ker(\rho_i)$, thus, by rank-nullity theorem, $Rank(\rho_i)=1$, then $\rho_i=\ket{v_i}\bra{v_i}$, and due to the inclusion of the kernels, $\{v\}^{\perp}\subset \{v_i\}^{\perp}$, so that, applying the orthogonal complement again, $\langle v_i\rangle\subset \langle v \rangle$, but $Tr(\rho_i)=Tr(\rho)=1$, then $\ket{v}=e^{i\theta}\ket{v_i}$, therefore $\rho=\rho_i$.
    \end{proof}

Some physical states in systems consisting of several subsystems can be identified by independent states in each subsystem, which is modeled mathematically as follows:
\begin{defi}
    Let $\mathcal{H}=\mathcal{H}_{A_1}\otimes\dots\otimes \mathcal{H}_{A_n}$ be the tensor product of Hilbert spaces describing a multipartite system. We will say that the state $\rho\in L(\mathcal{H})$ is \textbf{separable} if there exists a probability vector p and states $\rho_{j}^{A_i}\in L(\mathcal{H}_{A_j})$ such that:

\begin{equation}\rho=\sum_{j=1}^{m} p_j \, \bigotimes_{i\leq n} \rho_{j}^{A_i}.\end{equation}

\end{defi}

\begin{obs}\label{h}

If $\rho$ is a separable pure state, then for any j:

\begin{equation}\bigotimes_{i\leq n} \rho_{j}^{A_i}=\rho=\ket{v}\bra{v}.\end{equation}

\noindent Applying partial trace:

\begin{equation}
Tr_{A_{j\neq i}}^2(\rho)=(\rho_{j}^{A_i})^2=Tr_{A_{j\neq i}}(\rho^2)=Tr_{A_{j\neq i}}(\rho)=\rho_{j}^{A_i},\end{equation}

\noindent thus, $\rho_j^{A_i}$ is pure; hence, there exists $\ket{v_{A_i}}\in \mathcal{H}_{A_i}$ such that $\rho_j^{A_i}=\ket{v_{A_i}}\bra{v_{A_i}}$, therefore:

\begin{equation}(\bra{v}(\ket{v_{A_1}}\otimes\dots\otimes\ket{v_{A_n}}))\ket{v}=\ket{v_{A_1}}\otimes\dots \otimes \ket{v_{A_n}}:=\ket{v_{A_1},\dots,v_{A_{n}}},\end{equation}

\noindent in addition, since the norm is the vectors is 1, then $\vert\bra{v}(\ket{v_{A_1}}\otimes\dots\otimes\ket{v_{A_n}})\vert=1$.

\end{obs}

\begin{obs}
    It is important to note that not all states are separable, for example, the Bell state in $\mathcal{H}=\mathbb{C}^2\otimes \mathbb{C}^{2}$:
\begin{equation}\ket{\Psi}=\dfrac{1}{\sqrt{2}}(\ket{00}+\ket{11}),\end{equation}
\noindent  it is a non-separable or \textbf{entangled state}.

\end{obs}

\begin{prop}\label{prop 2.5.2}
    Let $\mathcal{H}=\mathcal{H}_A\otimes \mathcal{H}_B$ bipartite system. A pure state $\rho\in L(\mathcal{H})$is separable if and only if $Tr_{B}(\rho)^2=Tr_{B}(\rho)$
\end{prop}

\begin{proof}
    From observation \ref{h}, we know that if $\rho$ is separable pure state, then its partial traces are also pure. On the other hand, since $\rho$ is pure, there exists $\ket{v}$ normalized, such that $\rho=\ket{v}\bra{v}$. Let $\{\ket{u_i}\}_{i\leq d_A}$ and  $\{\ket{w_j}\}_{j\leq d_B}$ orthonormal basis of $\mathcal{H}_A$ and $\mathcal{H}_B$ respectively, then:
\begin{equation}\ket{v}=\sum_{\substack{i \leq d_A \\ j \leq d_B}} v_{ij} \, \ket{u_i,w_j}.\end{equation}
 \noindent By singular value decomposition \cite{treil2017linear}, exists $U\in U(d_A)$, $V\in U(d_B)$ and $\Sigma\in \mathbb{R}_{0}^{+\, d_A\times d_B}$ rectangular diagonal such that:
\begin{equation}(v_{ij})_{ij}=U\Sigma V^{*}.\end{equation}
\noindent If $r:=Rank(\Sigma)$, then:
\begin{equation}\ket{v}=\sum_{\substack{i \leq d_A \\ j \leq d_B\\ k\leq r}} \sigma_k U_{ik}V^{*}_{kj} \ket{u_i,w_j}=\sum_{k\leq r}\sigma_k \underbrace{\left(\sum_{i\leq d_A}U_{ik} \, \ket{u_i}\right)}_{:= \ket{\tilde{u_k}}} \otimes \underbrace{\left(\sum_{j\leq d_B}V^{*}_{kj} \, \ket{w_j}\right)}_{:=\ket{\tilde{w_k}}}=\sum_{k\leq r}\sigma_k\,\ket{\tilde{u_k},\tilde{w_k}}.\end{equation}

\newpage

The previous form of the vector is known as the \textbf{Schmidt decomposition} \cite{pathak2013elements}. Note that the $\{\ket{\tilde{u_k}}\}_{k\leq r}$ and $\{\ket{\tilde{w_k}}\}_{k\leq r}$ form orthonormal sets, as they are images of unitary operators, therefore, by completing bases, it follows that there exists for each vector a basis that leaves it in its Schmidt form. With this, we can write the state as:
\begin{equation}\label{2.67}\ket{v}\bra{v}=\sum_{i,j\leq r} \sigma_i \sigma_j \, \ket{\tilde{u_i}}\bra{\tilde{u_j}}\otimes \ket{\tilde{v_i}}\bra{\tilde{v_j}},\end{equation}
\noindent then, by hypothesis:
\begin{equation}Tr_{B}(\rho)=\sum_{i\leq r} \sigma_{i}^{2} \ket{\tilde{u_i}}\bra{\tilde{u_i}}\stackrel{!}{=}Tr_{B}^2(\rho)=\sum_{i\leq r} \sigma_i^{4}  \ket{\tilde{u_i}}\bra{\tilde{u_i}}
,\end{equation}
\noindent but since $Tr(\rho)=1$, the above happens only if there exists a unique non-zero $\sigma_k=1$. Therefore:
 \begin{equation}\rho=\ket{\tilde{u_k}}\bra{\tilde{u_k}}\otimes \ket{\tilde{w_k}}\bra{\tilde{w_k}}.\end{equation}
\end{proof}

\begin{obs}
    If we denote $\rho_A:=Tr_B(\rho)$ and $\rho_B:=Tr_A(\rho)$ \textbf{the reductions of \boldsymbol{\rho}}, then in general, by Schmidt decomposition in the vectors of the spectral decomposition:
\begin{equation}\rho=\sum_{k\leq d} p_k \, \ket{v_k}\bra{v_k}=\sum_{k\leq d} p_k \, \sum_{i,j\leq r_k}\sigma_{i,k}\sigma_{j,k} \ket{u_i^k}\bra{u_j^k}\otimes  \ket{w_i^k}\bra{w_j^k}, \end{equation}
\noindent thus:
\begin{equation}\rho_A=\sum_{k\leq d}\sum_{i\leq r_k} p_k\sigma^2_{i,k} \ket{u_i^k}\bra{u_i^k}\,\,\,\,\, ; \,\,\,\,\, \rho_B=\sum_{k\leq d}\sum_{i\leq r_k} p_k \sigma^2_{i,k} \ket{w_i^k}\bra{w_i^k},\end{equation}
\noindent therefore, $\rho_A$, $\rho_B$ are states in $\mathcal{H}_A$, $\mathcal{H}_B$ respectively. By induction, this is true for multipartite states.

\end{obs}

\newpage

\begin{obs}\label{p}
   Let's call $\gamma(\rho):=Tr(\rho^2)$ \textbf{the purity of \boldsymbol{\rho}}. If $\rho$ is pure then $\gamma(\rho)=Tr(\rho)=1$. On the other hand, if:
   \begin{equation}\gamma(\rho)=Tr\left(\sum_{k\leq d} p_k \ket{v_k}\bra{v_k}\right)^2=\sum_{k\leq d} p_k^2\stackrel{!}{=}1=\sum_{k\leq d} p_k,\end{equation}
\noindent   then exists a unique non-zero $p_j$ such that $p_j=1$, therefore, $\rho$ is pure.
\end{obs}

\begin{obs}
    
If $\rho$ is pure, from (\ref{2.67}) it follows that 
$\gamma(\rho_A)=\gamma(\rho_B)$, and by observation \ref{p}, $\rho$ is separable if and only if these purities are 1.
\end{obs}

\begin{prop}\label{ppp}
    Let $\mathcal{H}=\mathcal{H}_{A_1}\otimes\dots \otimes \mathcal{H}_{A_n}$ be a multipartite system and $\rho \in L(\mathcal{H})$ be a pure state. Then, $\rho$ is separable if and only if $\gamma(\rho_{A_1})=\gamma(\rho_{A_2})=\dots =\gamma(\rho_{A_{n-1}})=1$
\end{prop}

\begin{proof} By observation \ref{h}, we know that if $\rho$ is a pure separable state, then its reductions are pure, and therefore, their purities are $1$. On the other hand, by induction on the components of the system, we take the base case for a bipartite system (observation 2.5.5). Assuming the result holds for a system of $n-1$ parts, for a system of $n$ parts, we can treat it as a bipartite system with:
\begin{equation}\mathcal{H}=\mathcal{H}_{A_1}\otimes (\mathcal{H}_{A_2}\otimes \dots\otimes \mathcal{H}_{A_n}).\end{equation}
 \noindent Since $\gamma(\rho_{A_1})=1$, then $\rho=\rho_{A_1}\otimes \rho^{\prime}$, but $\gamma(\rho_{A_j})=\gamma(\rho^{\prime}_{A_j})$ for $j\neq 1$, thus, by inductive hypothesis $\rho^{\prime}$ is separable and therefore $\rho$ as well.

\end{proof}

\begin{obs}
\textbf{The purity of the reductions is a sufficient condition for the separability of mixed states}. Given the spectral decomposition $\rho = p_1 \ket{v_1}\bra{v_1} + \dots + p_d \ket{v_d}\bra{v_d}$,
if the reductions of \(\rho\) are pure, then analyzing the rank as in the proof of proposition \ref{pp}, it follows that the reductions of \(\ket{v_i}\bra{v_i}\) are also pure. Since this is a pure state, by proposition \ref{ppp}, this is equivalent to \(\ket{v_i}\bra{v_i}\) being separable, and therefore \(\rho\) is separable. However, \textbf{this is not a necessary condition}, as, for example, the state:
\begin{equation}\rho:=\dfrac{1}{2}\ket{00}\bra{00}+\dfrac{1}{2}\ket{11}\bra{11}=\dfrac{1}{2}\ket{0}\bra{0}\otimes \ket{0}\bra{0}+\dfrac{1}{2}\ket{1}\bra{1}\otimes \ket{1}\bra{1},\end{equation}
\noindent is a mixed separable state, with no pure reduction, since \(\rho_A = \rho_B = \frac{I}{2}\).
\end{obs}

\begin{obs}\label{obs 2.5.7}
    In observation \ref{p}, it was shown that the purity for any state is given by the sum of the squares of its spectrum. Since the spectrum is positive and sums to $1$, the purity must lie between $0$ and $1$. In the same observation, it was shown that $1$ is the maximum, which is achieved if and only if $\rho$ is pure. On the other hand, the global minimum can be obtained using Lagrange multipliers:
\begin{equation}\mathcal{L}(p_1,\dots,p_d,\lambda)=\sum_{k\leq d} p_k^2-\lambda \left(\sum_{k\leq d}p_k-1\right).\end{equation}
\noindent Optimizing, the following conditions are obtained:
\begin{equation}p_k=\dfrac{\lambda}{2}\,\,\,\,\,\, ; \,\,\,\,\, \sum_{k\leq d}p_k=1,\end{equation}
\noindent then $p_k=1/d$, Therefore, the minimum is achieved at \(\rho = \frac{I_{\mathcal{H}}}{d}\) with the value \(\gamma(\rho) = \frac{1}{d}\). This state is known as the \textbf{maximally mixed state}.

\end{obs}

\begin{obs}\label{obs 2.5.8}
    
Given that in pure states it is possible to characterize separability based on the purity of the reductions, this can be used as a measure of how entangled the state is. If the purities of the reductions are minimal, we will say that the state is \textbf{maximally entangled}. 

\end{obs}

\newpage
\section{Measures}\label{measures}

In this section we will introduce another definition of physical interest, which can be interpreted as a notion of generalized measurement in a system:

\begin{defi}
We say that the sequence of linear operators $(E_j)_{j\in J}$ is a Positive Operator-Valued Measure \textbf{(POVM)} if it satisfies $E_j=E_j^*\geq 0$ and:
\begin{equation}\sum_{j\in J} E_j=I_{\mathcal{H}}.\end{equation}
\noindent Additionally, given a state \(\rho\), then we denote:
\begin{equation}p_{j}(\rho):=\langle\rho,E_j\rangle,\end{equation}
\noindent where $\langle\rho,E_{j}\rangle=Tr(\rho^{*} \,E_j)$ is the Frobenius inner product in $L(\mathcal{H})$

\end{defi}
\begin{ex}
    Let $(\ket{f_j})_{j \in J} \subset \mathcal{H}$ be a \textbf{normalized tight frame} by observation \ref{o2.1.1}. Then $E_j:=\ket{f_j}\bra{f_j}$ form a \textbf{rank 1 POVM} with $p_j(\rho)=\bra{f_j}\rho\ket{f_j}$.
\end{ex}

\begin{ex}
An important example of a POVM is a Projective Valued Measure \textbf{(PVM)} which we denote as $\Pi_{j}$, which satisfies the additional condition:
\begin{equation}
\Pi_j \Pi_k = \delta_{j,k} \Pi_j.
\end{equation}


\end{ex}

\begin{obs}
    By spectral decomposition, the elements of a POVM can be written as: 
\begin{equation}
E_{j} = \sum_{k \leq n_j} p_{j,k} \ket{v_{j,k}}\bra{v_{j,k}}.
\end{equation}
\noindent Defining \(\ket{f_{j,k}} := \sqrt{p_{j,k}} \,\ket{v_{j,k}}\), we observe that for:
\begin{equation} J_c = \bigcup_{j \in J} \{(j,k)\}_{k \leq n_j},\end{equation}
\noindent \textbf{there is an associated canonical POVM of rank 1}:
\begin{equation}
E_{j,k} = \ket{f_{j,k}}\bra{f_{j,k}},
\end{equation}
\noindent corresponding to the normalized tight frame \((\ket{f_{j}})_{j \in J_c}\).
\end{obs}

POVMs can be connected to PVMs through the following result:

\begin{prop}[Naimark]\label{prop 2.6.1}
    Let \((E_j)_{j \in J}\) be a finite POVM in \(\mathcal{H}_A\). Then, there exists a Hilbert space \(\mathcal{H}_B\) such that, for every pure state \(\rho_B\) on \(\mathcal{H}_B\), there exists a PVM \((\pi_j)_{j \in J}\) in \(\mathcal{H}_B\) and a unitary operator \(U\in L(\mathcal{H}_A \otimes \mathcal{H}_B)\), so that the operators \(\Pi_{j} = U^*(I_{\mathcal{H}_A} \otimes \pi_j)U\) form a PVM in \(\mathcal{H}_A \otimes \mathcal{H}_B\) that satisfies:

\begin{equation}\langle \rho_{A}\otimes \rho_{B} ,\Pi_j\rangle=\langle\rho_A,E_j \rangle,\end{equation}

\noindent for any $\rho_{A}\in L(\mathcal{H}_A)$.

\end{prop}

\begin{proof}
   Consider \(\mathcal{H}_B = l^2(J)\), then,
 since \(E_j=E_j^{*} \geq 0\), we can define:
\begin{equation}A:=\sum_{j\in J} \sqrt{E_j}\otimes \ket{j},\end{equation}
\noindent then $A^{*}A=I_{\mathcal{H}_{A}}$, and $A^{*} (I_{\mathcal{H}}\otimes E_{j,j})A=E_j$. If \(\rho_B := \ket{u}\bra{u}\), then we can find a unitary operator \(U \in L(\mathcal{H}_A \otimes \mathcal{H}_B)\) such that:  
\begin{equation} U(I_{\mathcal{H}_{A}}\otimes \ket{u})=A,\end{equation}
\noindent for this, it suffices to define \(U\) on an orthonormal basis \(\{\ket{v_k}\}_k\) of \(\mathcal{H}_A\) as:  
\begin{equation}
U(\ket{v_a} \otimes \ket{u}) = A\ket{v_k}.
\end{equation}
Next, we take an orthonormal basis of \(\{\ket{u}\}^{\perp}\), say \(\{\ket{u_j}\}_j\), and define \(U(\ket{v_k} \otimes \ket{u_j})\) such that it forms an orthonormal basis of \(\{A\ket{v_k}\}_k^{\perp}\).  

\newpage

The above defines the PVM \(\Pi_{j} = U^*(I_{\mathcal{H}_A} \otimes \pi_j)U\) with $\pi_j=E_{j,j}$ and by the cyclicity of the trace:
$$\langle \rho_{A}\otimes \rho_{B} ,\Pi_j\rangle=\langle (I_{\mathcal{H}_A}\otimes \ket{u})\rho_{A} (I_{\mathcal{H}_A}\otimes \bra{u}),U^*(I_{\mathcal{H}_A} \otimes \pi_j)U\rangle$$
\begin{equation}=\langle \rho_A,A^{*}(I_{\mathcal{H_A}}\otimes E_{j,j})A)\rangle=\langle\rho_A,E_j \rangle.\end{equation}

\end{proof}

\begin{obs}
    If $\rho_{B}:=\ket{u}\bra{u}$, given that for every \(\rho_A \in L(\mathcal{H}_A)\), it holds that:
    \begin{equation}\langle\rho_A,(I_{\mathcal{H}_A}\otimes \bra{u})U^{*}(I_{\mathcal{H}_A}\otimes \pi_j)U(I_{\mathcal{H}_A}\otimes \ket{u})\rangle=\langle \rho_A,E_j\rangle,\end{equation}
\noindent then:
\begin{equation}E_j=(I_{\mathcal{H}_A}\otimes \bra{u})U^{*}(I_{\mathcal{H}_A}\otimes \pi_j)U(I_{\mathcal{H}_A}\otimes \ket{u}):=\bra{u}U^{*}(I_{\mathcal{H}_A}\otimes \pi_j)U\ket{u}.\end{equation}

\end{obs}

\newpage



\chapter{Thesis results}
\section{Bipartite separability in Harmonic Frames}\label{bipartite}
In this section we will analyze the bipartite separability condition for the states associated with the Harmonic frames. Let us first note that, by observation \ref{20}, a $\mathcal{G}$-Harmonic frame of dimension
$d$ is identified with a sequence $J$ of different elements of $\mathcal{G}$ whose length is $d$. Expressing $\mathcal{G}$ as a direct product of cyclic groups, through the isomorphism stated in observation \ref{21}, the frame $\Phi_J=(\xi\vert_{J})_{\xi \in \hat{G}}$ can be written as:
\begin{equation}\Phi_J=(\Psi(a_1^{r_0},\dots,a_m^{r_{m-1}}))\vert_{J})_{\substack{0\leq r_0\leq n_0-1\\ \dots \\ 0\leq r_{m-1}\leq n_{m-1}-1}}
=(\displaystyle \bigcirc_{j=0}^{m-1}\Psi_{j}(a_j^{r_j})\vert_{J_j})_{\substack{0\leq r_0\leq n_0-1\\ \dots \\ 0\leq r_{m-1}\leq n_{m-1}-1}},\end{equation}
\noindent where $\circ$ represents the Hadamard product, and $J_j$ is the projection of $J$ onto the $j$-th factor of the direct product. Since $J_j=(a_j^{s_{i,j}})_{0\leq i\leq d-1}$, the frame can be characterized from the matrix $S=(s_{i,j})$ which satisfies the restriction that $J$ has distinct elements only if no two rows are identical. Thus, the states associated with the frame can be written as:
$$\ket{v^{S}_{(r_j)_{j}}}:=\dfrac{1}{\sqrt{d}}\bigcirc_{j=0}^{m-1}\Psi_{j}(a_j^{r_j})\vert_{J_j}=\dfrac{1}{\sqrt{d}}\bigcirc_{j=0}^{m-1}(w_j^{r_j \,s_{i,j}})_{0\leq i\leq d-1}$$
\begin{equation}=\dfrac{1}{\sqrt{d}}\left(\prod_{j=0}^{m-1} w_{j}^{r_j \, s_{i,j}}\right)_{0\leq i\leq d-1}=\dfrac{1}{\sqrt{d}}\sum_{ 0\leq i\leq d-1} \prod_{j=0}^{m-1}  w_{j}^{r_j \, s_{i,j}} \, \ket{i},\end{equation}

\newpage


\noindent whose density matrix is of the form:
\begin{equation}\rho^{S}_{(r_j)_j}:=\ket{v^{S}_{(r_j)_{j}}}\bra{v^{S}_{(r_j)_{j}}}=\dfrac{1}{d}\sum_{ 0\leq i,k\leq d-1}\prod_{j=0}^{m-1}  w_j^{r_j \,(s_{i,j}-s_{k,j})}\, \ket{i}\bra{k}.\end{equation}
\noindent If $d=d_1\,d_2$ then $\mathbb{C}^{d\times d}\cong \mathbb{C}^{d_1\times d_1}\otimes \mathbb{C}^{d_2\times d_2}$ via  the Kronecker product $\ket{i}\to \ket{i_1,i_2}$, where $i_l\in \{0,\dots,d_l-1\}$ and $i=i_1 d_2+i_2$. With this, the state is written as:
\begin{equation}\label{3.4}\rho^{S}_{(r_j)_j}=\dfrac{1}{d}\sum_{\substack{0\leq i_1,k_1\leq d_1-1 \\ 0\leq i_2,k_2\leq d_2-1}}\prod_{j=0}^{m-1}  w_j^{r_j \,(s_{(i_1d_2+i_2),j}-s_{(k_1d_2+k_2),j})}\, \ket{i_1}\bra{k_1}\otimes \ket{i_2}\bra{k_2},\end{equation}
\noindent in this way, the reduction to $\mathbb{C}^{d_1\times d_1}$ is:
\begin{equation}\label{3.5}(\rho^{S}_{(r_j)_j})_{1}=\dfrac{1}{d}\sum_{\substack{0\leq i_1,k_1\leq d_1-1 \\ 0\leq i_2\leq d_2-1}}  \prod_{j=0}^{m-1}  w_j^{r_j \,(s_{(i_1d_2+i_2),j}-s_{(k_1d_2+i_2),j})}\, \ket{i_1}\bra{k_1},\end{equation}
\noindent then, its square is:
$$(\rho^{S}_{(r_j)_j})^2_{1}=\dfrac{1}{d^2}\sum_{\substack{0\leq i_1,k_1,\tilde{i_1},\tilde{k_1}\leq d_1-1 \\ 0\leq i_2,\tilde{i_2}\leq d_2-1}}  \prod_{j=0}^{m-1}  w_j^{r_j \,((s_{(i_1d_2+i_2),j}-s_{(k_1d_2+i_2),j})+(s_{(\tilde{i_1}d_2+\tilde{i_2}),j}-s_{(\tilde{k_1}d_2+\tilde{i_2}),j}))}\, \ket{i_1}\langle k_1 \vert \tilde{i_1} \rangle\bra{\tilde{k_1}}$$
\begin{equation}=\dfrac{1}{d^2}\sum_{\substack{0\leq i_1,k_1,\tilde{k_1}\leq d_1-1 \\ 0\leq i_2,\tilde{i_2}\leq d_2-1}}  \prod_{j=0}^{m-1}  w_j^{r_j \,((s_{(i_1d_2+i_2),j}-s_{(k_1d_2+i_2),j})+(s_{(k_1 d_2+\tilde{i_2}),j}-s_{(\tilde{k_1}d_2+\tilde{i_2}),j}))}\, \ket{i_1}\bra{\tilde{k_1}},\end{equation}
\noindent therefore, the purity is given by:
$$\gamma((\rho^{S}_{(r_j)_j})_{1})=\dfrac{1}{d^2}\sum_{\substack{0\leq i_1,k_1\leq d_1-1 \\ 0\leq i_2,\tilde{i_2}\leq d_2-1}}  \prod_{j=0}^{m-1}  w_j^{r_j \,((s_{(i_1d_2+i_2),j}-s_{(k_1d_2+i_2),j})+(s_{(k_1 d_2+\tilde{i_2}),j}-s_{(i_1d_2+\tilde{i_2}),j}))}$$
$$=\dfrac{1}{d^2}\left(\sum_{\substack{0\leq i_1,k_1\leq d_1-1 \\ 0\leq i_2=\tilde{i_2}\leq d_2-1}}1+\sum_{\substack{0\leq i_1=k_1\leq d_1-1 \\ 0\leq i_2\neq \tilde{i_2}\leq d_2-1}}1+\sum_{\substack{0\leq i_1\neq k_1\leq d_1-1 \\ 0\leq i_2\neq\tilde{i_2}\leq d_2-1}}  \prod_{j=0}^{m-1}  w_j^{r_j ((s_{(i_1d_2+i_2),j}-s_{(k_1d_2+i_2),j})+(s_{(k_1 d_2+\tilde{i_2}),j}-s_{(i_1d_2+\tilde{i_2}),j}))} \!\right)$$
\newpage
$$=\dfrac{1}{d^2}(d_1^2d_2+d_1d_2(d_2-1)+\dots$$ 
\begin{equation}\label{3.7}\left.\dots \, 4\sum_{\substack{0\leq i_1<k_1\leq d_1-1 \\ 0\leq i_2<\tilde{i_2}\leq d_2-1}}Cos\left(2\pi\left(\sum_{j=0}^{m-1} \dfrac{r_j\,((s_{(i_1d_2+i_2),j}-s_{(k_1d_2+i_2),j})+(s_{(k_1 d_2+\tilde{i_2}),j}-s_{(i_1d_2+\tilde{i_2}),j}))}{n_j}\right)\right)\right).\end{equation}
\noindent Let us note that the maximum value is achieved if and only if the cosines take the value of 1. Moreover, in such a case, we have that:
\begin{equation}\gamma((\rho^{S}_{(r_j)_j})_{1})=\dfrac{1}{d^2}\left(dd_1+d(d_2-1)+4\dfrac{d_1(d_1-1)d_2(d_2-1)}{4}\right)=\dfrac{1}{d}(d_1+d_2-1+(d_1-1)(d_2-1))=1,\end{equation}
\noindent thus, by the proposition \ref{prop 2.5.2}, we conclude the following propositions:

\begin{prop}\label{prop 3.1.1}
The Harmonic frames of dimension $d$ and $N=n_0\dots n_{m-1}$ elements can be characterized by a matrix $S=(s_{i,j})$ of $d\times m$ such that the rows are different and $s_{i,j}\in\{0,\dots,n_j-1\}$. Moreover, the states associated with the frame are of the form:    
\begin{equation}\ket{v^{S}_{(r_j)_{j}}}=\dfrac{1}{\sqrt{d}}\sum_{ 0\leq i\leq d-1} \prod_{j=0}^{m-1} w_{j}^{r_j \, s_{i,j}} \, \ket{i},\end{equation}
\noindent where $(r_j)_j$ are sequences such that $r_j\in\{0,\dots, n_j-1\}$. \textbf{For short, we will call this frame as \boldsymbol{(S,(n_j)_{j})}-Harmonic frame}.
\end{prop}

\begin{prop}\label{prop 3.1.2} Given a $(S,(n_j)_{j})$-Harmonic frame of dimension $d=d_1d_2$, then the state $\ket{v^{S}_{(r_j)_j}}$ is 
$(d_1,d_2)$-separable if and only if it is satisfied that:
\begin{equation}\label{3.10}\sum_{j=0}^{m-1} \dfrac{r_j\,((s_{(i_1d_2+i_2),j}-s_{(k_1d_2+i_2),j})+(s_{(k_1 d_2+\tilde{i_2}),j}-s_{(i_1d_2+\tilde{i_2}),j}))}{n_j}\in \mathbb{Z} \,\,\,\, : \, \substack{0\leq i_1<k_1\leq d_1-1 \\ 0 \leq i_2< \tilde{i_2}\leq d_2-1}.\end{equation}
\end{prop}

\newpage

\noindent On the other hand, let us note that the minimum value in (\ref{3.7}) occurs when all the cosines are equal to -1, which happens when:
\begin{equation}\sum_{j=0}^{m-1} \dfrac{r_j\,((s_{(i_1d_2+i_2),j}-s_{(k_1d_2+i_2),j})+(s_{(k_1 d_2+\tilde{i_2}),j}-s_{(i_1d_2+\tilde{i_2}),j}))}{n_j}\in \mathbb{Z}+\dfrac{1}{2} \,\,\,\, : \, \substack{0\leq i_1<k_1\leq d_1-1 \\ 0 \leq i_2< \tilde{i_2}\leq d_2-1}.\end{equation}
\noindent but in this case:
\begin{equation}\gamma((\rho^{J}_{(r_j)_j})_{1})=\dfrac{1}{d}(d_1+d_2-1-(d_1-1)(d_2-1))=\dfrac{2(d_1+d_2-1)-d}{d}.\end{equation}
\noindent Due to the symmetry under exchanging \( d_1 \) and \( d_2 \), we can assume wlog that \( d_2 = d_1 + k \) with \( k \geq 0 \). Therefore, \textbf{for the state to be maximally entangled, it must also satisfy} (remember observation \ref{obs 2.5.7} and \ref{obs 2.5.8}):
\begin{equation}2(d_1+(d_1+k)-1)-d_1(d_1+k)=1\Rightarrow d_1^{2}-d_1(4-k)-(2k-3)=0, \end{equation}

\noindent thus:

\begin{equation}d_{1\pm}=\dfrac{(4-k)\pm \sqrt{4+k^2}}{2}.\end{equation}

\noindent \textbf{The equation has integer solutions for \boldsymbol{k=0}, in which case \boldsymbol{d_{1-}=d_{2-}=1} and \boldsymbol{d_{1+}=d_{2+}=3} }. Otherwise, let us note that by interpreting the function as a positive real variable, the derivative is given by:

\begin{equation}\dfrac{d}{dk} \, d_{1\pm}=\dfrac{1}{2}\left(-1\pm \dfrac{k}{\sqrt{4+k^2}}\right),\end{equation}

\noindent then, for \( k \geq 1 \), \( d_{1_{\pm}} \) are strictly decreasing, thus, \( d_{1_{-}}(k) < 1 \), so that cannot be a positive integer. In the same way \( d_{1_{+}}(k) < 3 \), however, \( \lim_{k \to \infty} d_{1_{+}}(k) = 2 \), so \( d_{1_{+}}(k) \in (2,3) \), and therefore, it is not a positive integer either.\\

\newpage

Back to the proposition \ref{prop 3.1.2}, it can be noticed that if the condition (\ref{3.10}) is met, then it is also satisfied when \(i_1 > k_1\) (multiplying (\ref{3.10}) by \(-1\)). With this, we can rewrite the reductions in a simpler way, since:
 $$\prod_{j=0}^{m-1} w_j^{r_j \,(s_{(i_1d_2+i_2),j}-s_{(k_1d_2+i_2),j})}=Exp\left(2\pi i\left(\sum_{j=0}^{m-1} \dfrac{r_j\,(s_{(i_1d_2+i_2),j}-s_{(k_1d_2+i_2),j})}{n_j}\right)\right)$$
 \begin{equation}\label{3.16}=Exp\left(2\pi i\left(\sum_{j=0}^{m-1} \dfrac{r_j\,(s_{(i_1d_2+\tilde{i_2}),j}-s_{(k_1 d_2+\tilde{i_2}),j})}{n_j}\right)\right)=\prod_{j=0}^{m-1} w_j^{r_j \,(s_{(i_1d_2+\tilde{i_2}),j}-s_{(k_1 d_2+\tilde{i_2}),j})},\end{equation}
\noindent then, by taking $\tilde{i_2}=0$, the equation (\ref{3.5}) becomes:
\begin{equation}\label{3.17}(\rho^{S}_{(r_j)_j})_{1}=\dfrac{1}{d_1}\sum_{0 \leq i_1,k_1\leq d_1-1}  \prod_{j=0}^{m-1} w_j^{r_j \,(s_{i_1d_2,j}-s_{k_1d_2,j})} \,\ket{i_1}\bra{k_1}.\end{equation}
\noindent It can also be written (\ref{3.16}) as:
\begin{equation}\prod_{j=0}^{m-1} w_j^{r_j \,(s_{(i_1d_2+i_2),j}-s_{(i_1d_2+\tilde{i_2}),j})}=\prod_{j=0}^{m-1} w_j^{r_j \,(s_{(k_1d_2+i_2),j}-s_{(k_1d_2+\tilde{i_2}),j})},\end{equation}
\noindent thus, choosing $k_1=0$, from (\ref{3.4}) it can be obtained that the reduction to $\mathbb{C}^{d_2\times d_2}$ is:
$$(\rho^{S}_{(r_j)_j})_2=\dfrac{1}{d}\sum_{\substack{0\leq i_1\leq d_1-1 \\ 0\leq i_2,k_2\leq d_2-1}} \prod_{j=0}^{m-1} w_j^{r_j \,(s_{(i_1d_2+i_2),j}-s_{(i_1d_2+k_2),j})}\, \ket{i_2}\bra{k_2},$$
\begin{equation}\label{3.19}=\dfrac{1}{d_2}\sum_{0 \leq i_2,k_2\leq d_2-1} \prod_{j=0}^{m-1} w_j^{r_j \,(s_{i_2,j}-s_{k_2,j})}\, \ket{i_2}\bra{k_2},\end{equation}
\noindent and by observation \ref{h}, the state is written as $\rho^{S}_{(r_j)_j}=(\rho^{S}_{(r_j)_j})_{1}\otimes (\rho^{S}_{(r_j)_j})_{2}$.

\newpage

\section{Multipartite Entanglement in Harmonic Frames}\label{multipartite}

Suppose now that $d = d_1 d_2 \dots d_n$, we define:
\begin{equation}d_l^{\prime}=\prod_{k=l }^{n}d_{k}. \end{equation}
\noindent By induction $\mathbb{C}^{d\times d}\cong \mathbb{C}^{d_1\times d_1}\otimes \dots \otimes \mathbb{C}^{d_n\times d_n}$ via the Kronecker product, and by proposition \ref{ppp}, \( \rho^{S}_{(r_j)_j} \) will be separable if and only if the purities of its reductions to the systems \( C^{d_l \times d_l} \) for \( l \in \{1, \dots, n-1\} \) are equal to 1. 
We will show by induction the following proposition. 

\begin{prop}\label{prop 3.2.1}

Given a $(S,(n_j)_{j})$-Harmonic frame of dimension $d=d_1d_2\dots d_n$, then the state $\ket{v^{S}_{(r_j)_j}}$ is 
$(d_1,d_2,\dots,d_n)$-separable if and only if it is satisfied that:
\begin{equation}\label{3.21}\sum_{j=0}^{m-1} \dfrac{r_j\,((s_{(i_ld^{\prime}_{l+1}+\tilde{i}_{l+1}),j}-s_{(k_ld^{\prime}_{l+1}+\tilde{i}_{l+1}),j})+(s_{(k_l d^{\prime}_{l+1}+\tilde{k}_{l+1}),j}-s_{(i_ld^{\prime}_{l+1}+\tilde{k}_{l+1}),j}))}{n_j}\in \mathbb{Z} \,\,\,\, : \, \substack{0\leq i_l<k_l\leq d_l-1 \\ 0 \leq \tilde{i}_{l+1}< \tilde{k}_{l+1}\leq d_{l+1}^{\prime}-1\\ 1 \leq l\leq n-1},\end{equation}
\noindent and, in this case:
\begin{equation}(\rho^{S}_{(r_j)_j})_{l}=\dfrac{1}{d_l}\sum_{0 \leq i_l,k_l\leq d_l-1} \prod_{j=0}^{m-1} w_j^{r_j \,(s_{i_ld_{l+1}^{\prime},j}-s_{k_l d_{l+1}^{\prime},j})}\, \ket{i_l}\bra{k_l},\end{equation}
\noindent where $d_{n+1}^{\prime}:=1$ (empty product convention).
\end{prop}

\begin{proof} Indeed, we take as the base case \(n=2\), which was treated previously. Assuming it holds for \(n-1\) products, then for \(n\) products, taking \(d = d_1 d_2^{\prime}\), we have by proposition \ref{prop 3.1.2} that \(\gamma((\rho^{S}_{(r_j)_j})_{1}) = 1\) if and only if:
\begin{equation}\sum_{j=0}^{m-1} \dfrac{r_j\,((s_{(i_1d^{\prime}_{2}+\tilde{i}_2),j}-s_{(k_1d^{\prime}_{2}+\tilde{i}_2),j})+(s_{(k_1d^{\prime}_{2}+\tilde{k}_2),j}-s_{(i_1d^{\prime}_{2}+\tilde{k}_2),j}))}{n_j}\in \mathbb{Z} \,\,\,\, : \, \substack{0\leq i_1<k_1\leq d_1-1 \\ 0 \leq \tilde{i}_2< \tilde{k}_2\leq d_{2}^{\prime}-1},\end{equation}
\noindent furthermore, we know that of equations (\ref{3.17}) and (\ref{10}) that:
\begin{equation}(\rho^{S}_{(r_j)_j})_{1}=\dfrac{1}{d_1}\sum_{0 \leq i_1,k_1\leq d_1-1}   \prod_{j=0}^{m-1} w_j^{r_j \,(s_{i_1d_2^{\prime},j}-s_{k_1d_2^{\prime},j})} \,\ket{i_1}\bra{k_1},\end{equation}
\begin{equation}(\rho^{S}_{(r_j)_j})_{2^{\prime}}=\dfrac{1}{d_2^{\prime}}\sum_{0 \leq \tilde{i}_2,\tilde{k}_2\leq d_2^{\prime}-1} \prod_{j=0}^{m-1} w_j^{r_j \,(s_{\tilde{i}_2,j}-s_{\tilde{k}_2,j})}\, \ket{\tilde{i}_2}\bra{\tilde{k}_2},\end{equation}

\noindent but $(\rho^{S}_{(r_j)_j})_{2^{\prime}}$ has the same form as $\rho^{S}_{(r_j)_j}$ except for replacing $d$ with $d_2^{\prime}$. Since this resides in $\mathbb{C}^{d_2^{\prime}\times d_2^{\prime}}\cong \mathbb{C}^{d_2\times d_2}\otimes \dots\otimes \mathbb{C}^{d_n\times d_n}$, then by the inductive hypothesis, it is separable if and only if:
\begin{equation}\sum_{j=0}^{m-1} \dfrac{r_j\,((s_{(i_ld^{\prime}_{l+1}+\tilde{i}_{l+1}),j}-s_{(k_ld^{\prime}_{l+1}+\tilde{i}_{l+1}),j})+(s_{(k_l d^{\prime}_{l+1}+\tilde{k}_{l+1}),j}-s_{(i_ld^{\prime}_{l+1}+\tilde{k}_{l+1}),j}))}{n_j}\in \mathbb{Z} \,\,\,\, : \, \substack{0\leq i_l<k_l\leq d_l-1 \\ 0 \leq \tilde{i}_{l+1}< \tilde{k}_{l+1}\leq d_{l+1}^{\prime}-1\\ 2 \leq l\leq n-1},\end{equation}
\noindent and since $(\rho^{S}_{(r_j)_j})_{l}=((\rho^{S}_{(r_j)_j})_{2^{\prime}})_{l}$ for $l\geq 2$, it follows that:
\begin{equation}(\rho^{S}_{(r_j)_j})_{l}=\dfrac{1}{d_l}\sum_{0 \leq i_l,k_l\leq d_l-1} \prod_{j=0}^{m-1} w_j^{r_j \,(s_{i_ld_{l+1}^{\prime},j}-s_{k_l d_{l+1}^{\prime},j})}\, \ket{i_l}\bra{k_l}.\end{equation}
\end{proof}

\begin{obs}\label{obs 3.2.1}
Let's define: 
\begin{equation}\ket{v^{S}_{(r_j)}}_{l}=\dfrac{1}{\sqrt{d_l}}\sum_{0 \leq i_l\leq d_l-1} \prod_{j=0}^{m-1} w_j^{r_j \,(s_{i_ld_{l+1}^{\prime},j}+\theta_j)}\, \ket{i_l},\end{equation}
\noindent note that $(\rho^{S}_{(r_j)_j})_{l}=\ket{v^{S}_{(r_j)}}_{l}\bra{v^{J}_{(r_j)}}_{l}$, thus, by observation \ref{h} there exist $\theta_j$ such that:
\begin{equation}\ket{v^{S}_{(r_j)_j}}=\ket{(v^{S}_{(r_j)_j})_{1},\dots, (v^{S}_{(r_j)_j})_{n}},\end{equation}

\newpage
\noindent then, writing in Kronecker's product base:
$$\ket{v^{S}_{(r_j)_j}}=\dfrac{1}{\sqrt{d}}\sum_{ 0\leq i\leq d-1} \prod_{j=0}^{m-1} w_{j}^{r_j \, s_{i,j}} \, \ket{i}=\dfrac{1}{\sqrt{d}}\sum_{\substack{0\leq i_1\leq d_1-1\\ \dots \\ 0\leq i_n\leq d_{n}-1}}\prod_{j=0}^{m-1}w_{j}^{r_j s_{(i_1d_{2}^{\prime}+i_2d_{3}^{\prime}+\dots+i_{n}),j}} \ket{i_1,\dots,i_n}$$
\begin{equation}=\dfrac{1}{\sqrt{d}}\sum_{\substack{0\leq i_1\leq d_1-1\\ \dots \\ 0\leq i_n\leq d_{n}-1}} \prod_{j=0}^{m-1} w_{j}^{r_j(s_{i_1 d_{2}^{\prime
},j}+s_{i_{2}d_{3}^{\prime},j}+\dots+s_{i_n,j}+n\theta_j)}\ket{i_1,\dots,i_n},\end{equation}

\noindent thus, it is enough to take:
\begin{equation}\theta_j=\left(\dfrac{1-n}{n}\right) s_{0,j},\end{equation}
\noindent indeed, in this case:
$$\dfrac{1}{\sqrt{d}}\sum_{\substack{0\leq i_1\leq d_1-1\\ \dots \\ 0\leq i_n\leq d_{n}-1}} \prod_{j\leq m} w_{j}^{r_j(s_{i_1 d_{2}^{\prime
},j}+s_{i_{2}d_{3}^{\prime},j}+\dots+s_{i_n,j}+n\theta_j)}\ket{i_1,\dots,i_n}$$
$$=\dfrac{1}{\sqrt{d}}\sum_{\substack{0\leq i_1\leq d_1-1\\ \dots \\ 0\leq i_n\leq d_{n}-1}} \prod_{j=0}^{m-1} w_{j}^{r_j((s_{i_1 d_{2}^{\prime
},j}-s_{0,j})+(s_{i_2 d_{3}^{\prime
},j}-s_{0,j})\dots+(s_{i_{n-1} d_{n}^{\prime
},j}-s_{0,j})+s_{i_n,j})}\ket{i_1,\dots,i_n}$$
\begin{equation}=\dfrac{1}{\sqrt{d}}\sum_{\substack{0\leq i_1\leq d_1-1\\ \dots \\ 0\leq i_n\leq d_{n}-1}} \prod_{j=0}^{m-1} w_{j}^{r_j((s_{i_1 d_{2}^{\prime
},j}-s_{0,j})} \prod_{j=0}^{m-1} w_{j}^{r_j((s_{i_2 d_{3}^{\prime
},j}-s_{0,j})}\dots \prod_{j=0}^{m-1} w_{j}^{r_j s_{i_{n},j}}\ket{i_1,\dots,i_n},\end{equation}

\noindent but by the separability conditions (\ref{3.21}), we have that:
\begin{equation}\prod_{j=0}^{m-1} w_j^{r_j \,(s_{(i_ld^{\prime}_{l+1}+\tilde{i}_{l+1}),j}-s_{(k_ld^{\prime}_{l+1}+\tilde{i}_{l+1}),j})}=\prod_{j=0}^{m-1} w_j^{r_j \,(s_{(i_ld^{\prime}_{l+1}+\tilde{k}_{l+1}),j}-s_{(k_ld^{\prime}_{l+1}+\tilde{k}_{l+1}),j})} \,\,\,\, : \, \substack{0\leq i_l,k_l\leq d_l-1 \\ 0 \leq \tilde{i}_{l+1}, \tilde{k}_{l+1}\leq d_{l+1}^{\prime}-1\\ 1 \leq l\leq n-1}.\end{equation}
\noindent Choosing $k_l=\tilde{i}_{l+1}=0$ and:
\begin{equation}\tilde{k}_{l+1}=\sum_{r=l}^{n-1}i_{r+1}d^{\prime}_{r+2}\,\,\,\,\,\,\,\, \,\,\,\,\,\,\,\,k_{l}(i_r)_{r\geq l}:=\sum_{r=l}^{n}i_{r}d^{\prime}_{r+1},\end{equation}

\newpage
\noindent then, for all $l\leq n$ and $0\leq i_r\leq d_l-1$:
\begin{equation}\label{3.35}\prod_{j=0}^{m-1} w_j^{r_j \,(s_{i_ld^{\prime}_{l+1},j}-s_{0,j})}=\prod_{j=0}^{m-1} w_j^{r_j \,(s_{k_{l}(i_r)_{r\geq l},j}-s_{k_{l+1}(i_r)_{r\geq l+1},j})},\end{equation}

\noindent where $k_{n+1}(i_r)_{r\geq n+1}=0$ (empty sum convention). Applying telescoping summation to the exponents, we are left with the following:

$$\dfrac{1}{\sqrt{d}}\sum_{\substack{0\leq i_1\leq d_1-1\\ \dots \\ 0\leq i_n\leq d_{n}-1}} \prod_{j=0}^{m-1} w_{j}^{r_j(s_{i_1 d_{2}^{\prime
},j}+s_{i_{2}d_{3}^{\prime},j}+\dots+s_{i_n,j}+n\theta_j)}\ket{i_1,\dots,i_n}$$
$$=\dfrac{1}{\sqrt{d}}\sum_{\substack{0\leq i_1\leq d_1-1\\ \dots \\ 0\leq i_n\leq d_{n}-1}} \prod_{j=0}^{m-1} w_j^{r_j \,(s_{k_{1}(i_r)_{r\geq 1},j}-s_{i_{n},j})} \prod_{j=0}^{m-1} w_{j}^{r_j s_{i_{n},j}}\ket{i_1,\dots,i_n}$$
\begin{equation}=\dfrac{1}{\sqrt{d}}\sum_{\substack{0\leq i_1\leq d_1-1\\ \dots \\ 0\leq i_n\leq d_{n}-1}}\prod_{j=0}^{m-1} w_{j}^{r_j s_{(i_1d_{2}^{\prime}+i_2d_{3}^{\prime}+\dots+i_{n}),j}} \ket{i_1,\dots,i_n}=\ket{v^{S}_{(r_j)_j}},\end{equation}

\noindent therefore, we can write the factors as:
\begin{equation}\label{3.37}\ket{v^{S}_{(r_j)}}_{l}=\dfrac{1}{\sqrt{d_l}}\sum_{0 \leq i_l\leq d_l-1} \prod_{j=0}^{m-1} w_j^{r_j \,(s_{id_{l+1}^{\prime},j}+(\frac{1-n}{n})s_{0,j})}\, \ket{i_l}.\end{equation}

\begin{obs}\label{obs 3.2.2}

 With the above observation, and taking into account that conditions (\ref{3.35}) (for all $i_r$ and $l$) imply by themselves the separability of $\ket{v^{S}_{(r_j)_j}}$, then it is separable if and only:
\begin{equation}\prod_{j=0}^{m-1} w_j^{r_j \,(s_{i_ld^{\prime}_{l+1},j}-s_{0,j})}=\prod_{j=0}^{m-1} w_j^{r_j \,(s_{k_{l}(i_r)_{r\geq l},j}-s_{k_{l+1}(i_r)_{r\geq l+1},j})} \,\,\,\, : \, \substack{0\leq i_r\leq d_r-1 \\  1 \leq l\leq n},\end{equation}

\noindent or in other words:
\begin{equation}\sum_{j=0}^{m-1} \dfrac{r_j\,((s_{i_ld^{\prime}_{l+1},j}-s_{0,j})+(s_{k_{l+1}(i_r)_{r\geq l+1},j}-s_{k_{l}(i_r)_{r\geq l},j}))}{n_j}\in \mathbb{Z} \,\,\,\, : \, \substack{0\leq i_r\leq d_r-1 \\  1 \leq l\leq n}.\end{equation}

\end{obs}

\newpage

\end{obs}

\section{Designing \boldsymbol{C_N}-Frames}\label{Cn frames}

From observation 2.4.4, we have that the character table of the group \( C_N \) given by $(w_{N}^{ij})_{0\leq i,j\leq n-1}$ (with $w_N=e^{2\pi i/N}$). Normalizing, we obtain the Fourier matrix:

\begin{equation}
F_N = \frac{1}{\sqrt{N}} \begin{pmatrix}
1 & 1 & 1 & \cdots & 1 \\
1 & \omega_N & \omega_N^2 & \cdots & \omega_N^{N-1} \\
1 & \omega_N^2 & \omega_N^4 & \cdots & \omega_N^{2(N-1)} \\
\vdots & \vdots & \vdots & \ddots & \vdots \\
1 & \omega_N^{N-1} & \omega_N^{2(N-1)} & \cdots & \omega_N^{(N-1)(N-1)}
\end{pmatrix}.
\end{equation}  

\noindent By the orthogonality relations discussed in observation \ref{xd}, this matrix is unitary, so it can be useful for implementing the normalized \( C_N \) frames as POVMs via Naimark's theorem. Indeed, consider a \( C_N \) frame of dimension \( d \), since $m=1$, then by the proposition \ref{prop 3.1.1},  we can write $S=(s_i)_{0\leq i\leq d-1}$, where the \( s_i \)'s are distinct in $\{0,\dots, N-1\}$, thus, taking \( C \in \{\max\limits_{0\leq i\leq d-1} s_id, \ldots, Nd-1\} \), we can define the permutation \( \sigma_{S}(iN)=C-s_id \). Consider the permutation operator \( P_{\sigma_{S}} \) then, we define: 
\begin{equation}U_{S}:=F_{dN}P_{\sigma_{S}} \in \mathbb{C}^{dN\times dN}\cong \mathbb{C}^{d\times d}\otimes \mathbb{C}^{N\times N}.\end{equation}

\noindent Note that in the computational basis
$$U_{S}=\dfrac{1}{\sqrt{dN}}\sum_{\substack{0\leq i,j\leq dN-1}}w_{dN}^{ij}\, \ket{i}\bra{\sigma_{S}^{-1}(j)}=\dfrac{1}{\sqrt{dN}}\sum_{\substack{0\leq i,j\leq dN-1}}w_{dN}^{i\sigma_{S}(j)}\, \ket{i}\bra{j}$$
\begin{equation}=\dfrac{1}{\sqrt{dN}}\sum_{\substack{0\leq i_1,j_1\leq d-1 \\ 0\leq i_2,j_2\leq N-1}}w_{dN}^{(i_1N+i_2)\sigma_{S}(j_1N+j_2)}\ket{i_1}\bra{j_1}\otimes \ket{i_2}\bra{j_2},\end{equation}

\newpage
\noindent then:
$$U_{S}^{*}(I_{d}\otimes \ket{k}\bra{k})U_{S}$$
$$=\dfrac{1}{dN}\sum_{\substack{0\leq i_1,j_1,\tilde{i}_1,\tilde{j}_1\leq d-1 \\ 0\leq i_2,j_2,\tilde{i}_2,\tilde{j}_2\leq N-1}}w_{Nd}^{(\tilde{i}_1N+\tilde{i}_2)\sigma_{S}(\tilde{j}_1 N+\tilde{j}_2)-(i_1N+i_2)\sigma_{S}(j_1N+j_2)}(\ket{j_1}\bra{i_1}\otimes \ket{j_2}\bra{i_2})(I_{d}\otimes \ket{k}\bra{k})(\ket{\tilde{i}_1}\bra{\tilde{j}_1}\otimes \ket{\tilde{i}_2}\bra{\tilde{j}_2})$$
\begin{equation}=\dfrac{1}{dN}\sum_{\substack{0\leq i_1,j_1,\tilde{j}_1\leq d-1 \\ 0\leq j_2,\tilde{j}_2\leq N-1}}w_{Nd}^{(i_1N+k)\sigma_{S}(\tilde{j}_1N+\tilde{j}_2)-(i_1N+k)\sigma_{S}(j_1N+j_2)}\ket{j_1}\bra{\tilde{j}_1}\otimes\ket{j_2}\bra{\tilde{j}_2}.\end{equation}

\noindent Therefore:
$$\bra{0}U^{*}_{S}(I_{d}\otimes \ket{k}\bra{k})U_{S}\ket{0}=\dfrac{1}{dN}\sum_{0\leq i_1,j_1,\tilde{j}_1\leq d-1} w_{Nd}^{(i_1N+k)(\sigma_{S}(\tilde{j}_1N)-\sigma_{S}(j_1 N))}\ket{j_1}\bra{\tilde{j}_1}$$
$$=\dfrac{1}{dN}\sum_{0\leq i_1,j_1,\tilde{j}_1\leq d-1} w_{Nd}^{(i_1N+k)((C-d s_{\tilde{j}_1}-(C-d s_{j_1}))}\ket{j_1}\bra{\tilde{j}_1}=\dfrac{1}{N}\sum_{0\leq j_1,\tilde{j}_1\leq d-1} w_{N}^{k (s_{j_1}- s_{\tilde{j}_1})}\ket{j_1}\bra{\tilde{j}_1}$$
\begin{equation}\label{3.44}=\dfrac{d}{N}\ket{v^{S}_{k}}\bra{v^{S}_k},\end{equation}

\noindent which is the element $k$ of the POVM associated to the normalization of the $C_N$-frame, since, recalling the proposition \ref{3}, we have that the frame constant $A$ is:
\begin{equation}A=\dfrac{1}{d}\sum_{0\leq k\leq N-1} \vert\vert v_{k}^{S}\vert\vert^2=N/d .\end{equation}

\begin{obs}
    Note that, from (\ref{3.44}) it can be observed that the condition on the permutation $\sigma$ can be further relaxed, since it will form the desired POVM if and only if:

\begin{equation}\dfrac{1}{d}\sum_{0\leq i\leq d-1}w_{Nd}^{(iN+k)(\sigma(\tilde{j}N)-\sigma(jN))}= w_{N}^{k (s_{j}- s_{\tilde{j}})} \,\,\,\, : \, \substack{0\leq j,\tilde{j}\leq d-1 \\ 0 \leq k\leq N-1},\end{equation}

\newpage

\noindent but this implies that:
\begin{equation}\left|\sum_{0\leq i\leq d-1}w_{Nd}^{iN(\sigma(\tilde{j}N)-\sigma(jN))}\right|=d,\end{equation}
\noindent which occurs only if $d\vert\sigma(\tilde{j}N)-\sigma(jN)$, then the condition becomes:
\begin{equation}w_{N}^{k(\sigma(\tilde{j}N)-\sigma(jN))/d}= w_{N}^{k (s_{j}- s_{\tilde{j}})} \,\,\,\, : \, \substack{0\leq j,\tilde{j}\leq d-1 \\ 0 \leq k\leq N-1},\end{equation}
\noindent what we have left in:
\begin{equation}\dfrac{1}{N}\left(\dfrac{\sigma(\tilde{j}N)-\sigma(jN)}{d}+(s_{\tilde{j}}-s_{j})\right)\in \mathbb{Z} \,\,\,\, : \, 0\leq j,\tilde{j}\leq d-1,\end{equation}
\noindent but:
\begin{equation}\left(\dfrac{\sigma(\tilde{j}N)-\sigma(jN)}{d}+(s_{\tilde{j}}-s_{j})\right)\in \{-2(N-1),\dots,2(N-1)\},\end{equation}
\noindent then:
\begin{equation}\left(\dfrac{\sigma(\tilde{j}N)-\sigma(jN)}{d}+(s_{\tilde{j}}-s_{j})\right)\in \{-N,0,N\} \,\,\,\, : \, 0\leq j,\tilde{j}\leq d-1,\end{equation}
\noindent thus:
\begin{equation}\sigma(\tilde{j}N)-\sigma(jN)=d(N_{\tilde{j},j}+(s_j-s_{\tilde{j}})) \,\,\,\, : N_{\tilde{j},j}\in \{-N,0,N\}.\end{equation}
\noindent With this $\sigma(\tilde{j}N)=\sigma(0)+d(N_{\tilde{j},0}+(s_0-s_{\tilde{j}}))$, and therefore $N_{\tilde{j},j}=N_{\tilde{j},0}-N_{j,0}$, so, renaming $N_j:=N_{j,0}$, the permutation can be written as:
\begin{equation}\sigma(jN)=\sigma(0)+d(N_{j}-(s_j-s_{0})) \,\,\,\, : N_{j}\in \{-N,0,N\},\end{equation}
\noindent where $N_{0}=0$. Wlog, we can choose a strictly increasing $s_i$,  which rules out that $N_j=-N$, because $\sigma(jN)\in \{0,\dots,Nd-1\}$, then $N_j\in\{0,N\}$. Note that in this case $\sigma$ is already injective in multiplies of N, since if $\sigma(\tilde{j}N)=\sigma(jN)$, then $(N_{\tilde{j}}-N_j)=(s_{\tilde{j}}-s_j)$, but if wlog $\tilde{j}>j$ then, $0<N_{\tilde{j}}-N_j<N$, which is a contradiction. The last thing necessary for it to be a valid permutation is that it is in the corresponding range of values, that is to say:

\begin{equation}d(s_j-s_{0})\leq \sigma(0)+dN_{j}\leq d(N+s_j-s_0)-1 \,\,\,\, : 0\leq j\leq d-1,\end{equation}

\noindent thus, to choose some $N_j=N$ it must be satisfied that $\sigma(0)\leq d(s_j-s_0)-1$, and if $N_j=0$, then $\sigma(0)\geq d(s_j-s_0)$, therefore, if m is the smallest element such that $N_m=N$ and M is the largest such that $N_M=0$, then:

\begin{equation}d(s_M-s_0)\leq \sigma(0)\leq d(s_{m}-s_0)-1.\end{equation}

\end{obs}

\begin{obs}
The above observation gives the forms of permutations that build the POVM in order, however, it is still valid to build the POVM by permuting the elements. This increases even more the number of possible permutations, because in this case $\sigma$  will design the POVM, if and only if there exists a permutation $\sigma^{\prime}$ such that:

\begin{equation}\dfrac{1}{d}\sum_{0\leq i\leq d-1}w_{Nd}^{(iN+k)(\sigma(\tilde{j}N)-\sigma(jN))}= w_{N}^{\sigma^{\prime}(k) (s_{j}- s_{\tilde{j}})} \,\,\,\, : \, \substack{0\leq j,\tilde{j}\leq d-1 \\ 0 \leq k\leq N-1}.\end{equation}

\noindent Proceeding as before, we have again that $d\vert\sigma(\tilde{j}N)-\sigma(jN)$, and the condition becomes:  

\begin{equation}\dfrac{1}{N}\left(k\left(\dfrac{\sigma(\tilde{j}N)-\sigma(jN)}{d}\right)+\sigma^{\prime}(k)(s_{\tilde{j}}-s_{j})\right)\in \mathbb{Z}\,\,\,\, : \, \substack{0\leq j<\tilde{j}\leq d-1 \\ 0 \leq k\leq N-1},\end{equation}

\noindent where restrictions were eliminated by noting that the conditions are trivially satisfied for $j=\tilde{j}$, and that if they are satisfied for $j<\tilde{j}$, then they will also be satisfied for $j>\tilde{j}$ via factoring a $-1$.
    
\end{obs}
From what has been analyzed in this section, we can conclude with the following proposition:

\newpage

\begin{prop}\label{prop 3.3.1}
    Given a $C_N$ frame of dimension d characterized by $S=(s_i)_{0\leq i\leq d-1}$, then, the associated POVM given by $E_{k}=\frac{d}{N}\ket{v_{k}^{S}}\bra{v_{k}^{S}}$ it can be implemented via Naimark's theorem by applying the matrix $U_{S}=F_{dN}P_{\sigma_{S}}$, where $F_{dN}$ is the Fourier matrix of dimension $dN$ and $P_{\sigma_{S}}$ the permutation operator associated with $\sigma_{S}$, which in case of satisfying:
\begin{equation}d\vert\sigma_{S}(\tilde{j}N)-\sigma_{S}(jN)\,\,\,\,\,\, \wedge \,\,\,\,\, N\bigg\vert\left(k\left(\dfrac{\sigma_S(\tilde{j}N)-\sigma_S(jN)}{d}\right)+\sigma^{\prime}(k)(s_{\tilde{j}}-s_{j})\right)\,\,\,\, : \, \substack{0\leq j<\tilde{j}\leq d-1 \\ 0 \leq k\leq N-1}.\end{equation}
\noindent for some permutation $\sigma^{\prime}$, then it follows that:
\begin{equation}
E_{\sigma^{\prime}(k)}=\bra{0} U^{*}_{S}(I_{d}\otimes \ket{k}\bra{k})U_{S}\ket{0},    
\end{equation}
\noindent  which can be represented graphically by the following circuit:

\vspace{0.2cm}

\begin{figure}[hbtp]
    \centering
    \includegraphics[scale=0.18]{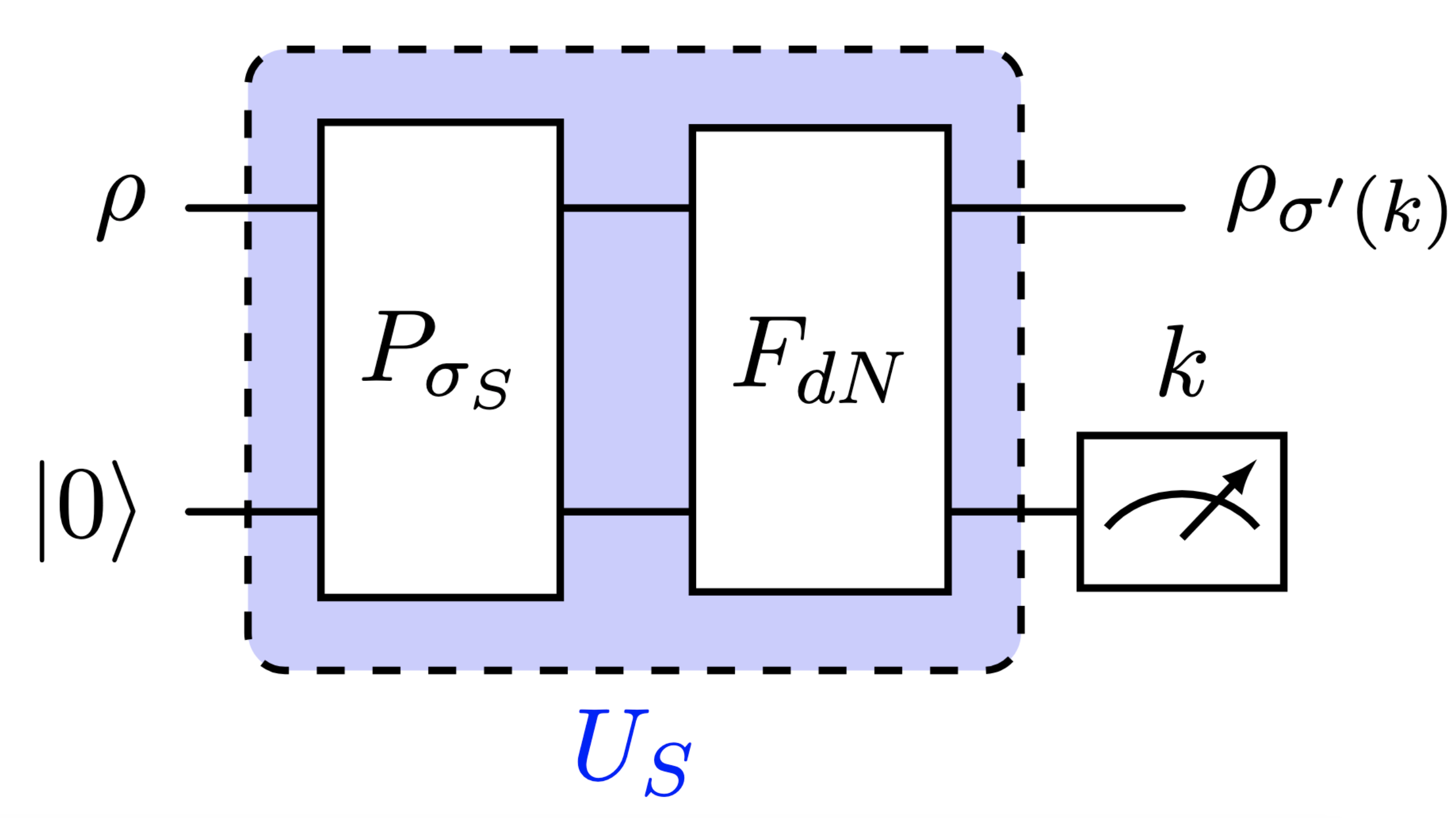}
    \caption{ \doublespacing
 We show the implementation of the d dimensional POVM $E_{\sigma^{\prime}(k)}$ by considering ancillary system of dimension N prepared $\ket{0}$ and the gate $U_S=F_{dN}P_{\sigma_{S}}$. In this case, the effect of measuring with the computational basis in the ancillary system is equivalent to measuring with the POVM associated with the $C_N$ frame permuted by $\sigma^{\prime}$.}
    \label{fig:enter-label}
\end{figure}

\end{prop}

\begin{ex}
    Consider the $C_N$ frame of dimension d with $s_j=j$, then this frame can be implemented by $\sigma_{S}(jN)=jd$. Indeed, let us note that:
\begin{equation}k\left(\dfrac{\sigma_{S}(\tilde{j}N)-\sigma_{S}(jN)}{d}\right)+\sigma^{\prime}(k)(s_{\tilde{j}}-s_{j})=(k+\sigma^{\prime}(k))(\tilde{j}-j),\end{equation}
    \noindent then, this element is divisible by N for all j, $\tilde{j}$ and k if and only if:
     \begin{equation}\sigma^{\prime}(k)=\left\{ \begin{array}{lcc}
             0\,\,\,\,\,\,\,\,\,\,\,\,\, :k=0      \\    
                 N-k\,\, :k\neq 0
             
             \end{array} \right. .\end{equation}
\end{ex}

\section{Simple example of Harmonic POVM on qubits computer}\label{example povm}

In a quantum computer, the dimension of the system we wish to measure will be $d=2^{\mathcal{N}}$ with $\mathcal{N}$ the number of qubits. The simplest case is when $\mathcal{N}=1$, furthermore, $N=2^{M}\geq d$ so that gates can be applied, so the simplest and most non-trivial case is when $M=2$. Consider $s_0=0, s_1=1$, then the vectors are:
\begin{equation}\ket{v^{S}_k}=\dfrac{1}{2} \begin{pmatrix}
    1\\
    w_4^{k}
\end{pmatrix},\end{equation}
\noindent with $w_4$ the fourth roots of unity. The POVM will then be:
\begin{equation}\label{3.61}E_k=\dfrac{1}{4}\begin{pmatrix}
    1 & w_4^{-k} \\

    w_4^{k} & 1
\end{pmatrix},\end{equation}
\noindent it is proven to be a POVM because the sum of all unit roots is $0$. In this case the gates we must construct are $F_8$ and $P_{\sigma_{S}}$, where $\sigma_{S}$ will implement in the order $\sigma^{\prime}$ if it satisfies for all $k\in\{0,3\}$:
\begin{equation}2\vert\sigma_S(4)-\sigma_S(0)\,\,\,\,\,\, \wedge \,\,\,\,\, 4\bigg\vert\left(k\left(\dfrac{\sigma_S(4)-\sigma_S(0)}{2}\right)+\sigma^{\prime}(k)\right).\end{equation}
\noindent To implement $F_8$, you can use the circuit shown in \cite{nielsen2010quantum}. On the other hand, we can choose $\sigma_{S}=(46)(57)$, thus, $\sigma_{S}(4)-\sigma_{S}(0)=6$ and $4\vert (3k+\sigma^{\prime}(k))$ for all k if and only if $\sigma^{\prime}(k)=k$, but this permutation can be implemented by means of $CNOT_{0\to 1}$.




\begin{figure}[hbtp]
    \centering
    \begin{quantikz}[row sep=0.8cm]
\lstick{$\rho$}& \ctrl{1}\gategroup[2,steps=1,style={dashed,rounded
corners,fill=blue!20, inner
xsep=2pt},background,label style={label
position=below,anchor=north,yshift=-0.2cm}]{{\sc
$\textcolor{blue}{P_{\sigma_{S}}}$}}&\gate{H} \gategroup[3,steps=7,style={dashed,rounded
corners,fill=red!20, inner
xsep=2pt},background,label style={label
position=below,anchor=north,yshift=-0.2cm}]{{\sc
$\textcolor{red}{F_{8}}$}}& \gate{S}& \gate{T}&&&&\swap{2}& &\rstick{$\rho^{S}_{2i+j}$}\\
\lstick{$\ket{0}$}& \gate{X}&&\ctrl{-1}  & &\gate{H}&\gate{S}&&&&\meter{i} \\
\lstick{$\ket{0}$}&  &&&\ctrl{-2}&&\ctrl{-1}&\gate{H}&\targX{}&&\meter{j}
\end{quantikz}
    \caption{\doublespacing We show the implementation of the POVM (\ref{3.61}) using the proposition \ref{prop 3.3.1}. In this case, two ancilla qubit were used, the gate $P_{\sigma_{S}}=CNOT_{0\to 1}$, and $F_8$, whose construction is described in \cite{nielsen2010quantum}.}
    \label{fig:enter-label}
\end{figure}
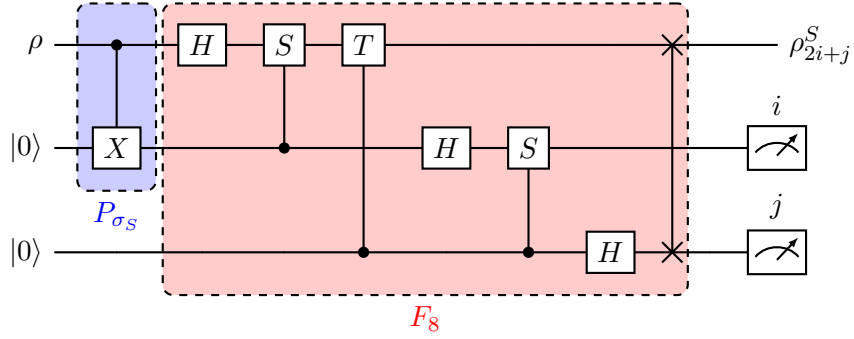

\section{Simple example of Harmonic state on quantum computer}\label{example separability}

To illustrate separability in a qubit computer (remember propositions \ref{prop 3.1.1} and \ref{prop 3.1.2}), the easiest case we can consider would be when $m=1$ ($C_N$ frame), with $n=2$ (Bipartite), $d_1=d_2=2$ and thus $d=4$, so the smallest $N$ is $4$, therefore $s_i=i$ is the only one that fulfills the cardinality condition. The states are of the form:
\begin{equation}\ket{v^{S}_{k}}=\dfrac{1}{2}\begin{pmatrix}
    1\\
    w_4^{k}\\
    w_4^{2k}\\
    w_{4}^{3k}
\end{pmatrix}.\end{equation}
\noindent In this case, the separability condition for the state vector $\ket{v^{S}_k}$ is:
\begin{equation} 4\vert k\underbrace{((s_{0}-s_{2})+(s_{3}-s_{1}))}_{=0},\end{equation}
\noindent so all the vectors are separable, furthermore $\theta=0$ (remember \ref{obs 3.2.1}), then by (\ref{3.37}):
\begin{equation}\ket{v^{S}_{k}}_{1}\otimes\ket{v^{S}_{k}}_{2}=\dfrac{1}{\sqrt{2}}\begin{pmatrix}
    1\\
    w_{4}^{2k}
\end{pmatrix}\otimes \dfrac{1}{\sqrt{2}}\begin{pmatrix}
    1\\
    w_{4}^{k}
\end{pmatrix}=\dfrac{1}{2}\begin{pmatrix}
    1\\
    w_4^{k}\\
    w_4^{2k}\\
    w_{4}^{3k}
\end{pmatrix}=\ket{v^{S}_{k}}.\end{equation}
\noindent We can increase the difficulty by taking $m=1$,$n=2$, $d_1=d_2=2$, $N\geq 4$ and any $S$ with the respective dimensions. In this case the condition of separability is:
\begin{equation}N\vert k((s_{0}-s_{2})+(s_{3}-s_{1})).\end{equation}
\noindent If we suppose that some vector indexed by k satisfies it, then, as $\theta=-s_0/2$:
\begin{equation}\ket{v^{S}_{k}}_{1}\otimes\ket{v^{S}_{k}}_{2}=\dfrac{1}{\sqrt{2}}\begin{pmatrix}
    w_{N}^{k \frac{s_0}{2}}\\
    w_{N}^{k (s_{2}-\frac{s_0}{2})}
\end{pmatrix}\otimes \dfrac{1}{\sqrt{2}}\begin{pmatrix}
    w_{N}^{k \frac{s_{0}}{2}}\\
    w_{N}^{k (s_{1}-\frac{s_0}{2})}
\end{pmatrix}=\dfrac{1}{2}\begin{pmatrix}
    w_{N}^{k s_0}\\
    w_{N}^{k s_1}\\
    w_{N}^{k s_2}\\
    w_{N}^{k (s_2+s_1-s_0)}
\end{pmatrix}=\dfrac{1}{2}\begin{pmatrix}
    w_{N}^{k s_0}\\
    w_{N}^{k s_1}\\
    w_{N}^{k s_2}\\
    w_{N}^{k s_3}
\end{pmatrix}=\ket{v^{S}_k},\end{equation}
\noindent since $k(s_2+s_1-s_0)=k s_3$ (mod $N$).


\chapter{Conclusions}
\section{Summary of Results
}\label{summary}
The presented thesis aimed to study the construction and implementation of tight frames in quantum computers as states (section \ref{state}) and as POVM (section \ref{measures}). Using representation theory (section \ref{representation}), group frames were introduced (section \ref{group frames}), and due to the known structure of abelian groups, efforts were focused on studying Harmonic frames (section \ref{harmonic}). In the first section of the results we analyzed the bipartite separability of the harmonic frame (section \ref{bipartite}). Specifically, the Harmonic frames of $N$ elements, starting from abelian groups $G=C_{n_0}\times\dots \times C_{n_{m-1}}$ such that $n_0 n_1...n_{m-1}=N$, and a sequence $J$ of $G$ with $d\leq N$ distinct elements (the dimension of the frame) was characterized by the proposition \ref{prop 3.1.1}, in which it was taken $J_j=(a_j^{s_{i,j}})_{0\leq i\leq d-1}$ the sequence of the $j$-th components of $J$, what defines the matrix $S:=(s_{i,j})$ of size $d\times m$, which satisfies the property of having distinct rows and that its elements $s_{i,j}\in \{0,\dots,n_j-1\}$. With this, it was obtained that the states of the frame written in the computational base are of the form:
\begin{equation}\label{4.1}\ket{v^{S}_{(r_j)_{j}}}=\dfrac{1}{\sqrt{d}}\sum_{ 0\leq i\leq d-1} \prod_{j=0}^{m-1} w_{}^{r_j \, s_{i,j}} \, \ket{i}\end{equation}
\noindent indexed by a sequence $(r_j)_j$ such that $r_j\in \{0,\dots, n_j-1\}$. Thanks to this, it was possible to obtain the following proposition \ref{prop 3.1.2}, which gives a necessary and sufficient condition for the bipartite separability of these states, given by:
\begin{equation}\sum_{j=0}^{m-1} \dfrac{r_j\,((s_{(i_1d_2+i_2),j}-s_{(k_1d_2+i_2),j})+(s_{(k_1 d_2+\tilde{i_2}),j}-s_{(i_1d_2+\tilde{i_2}),j}))}{n_j}\in \mathbb{Z} \,\,\,\, : \, \substack{0\leq i_1<k_1\leq d_1-1 \\ 0 \leq i_2< \tilde{i_2}\leq d_2-1}.\end{equation}
\noindent The possibility of maximally intertwined states was also studied, in which case the necessary condition was that:
\begin{equation}\sum_{j=0}^{m-1} \dfrac{r_j\,((s_{(i_1d_2+i_2),j}-s_{(k_1d_2+i_2),j})+(s_{(k_1 d_2+\tilde{i_2}),j}-s_{(i_1d_2+\tilde{i_2}),j}))}{n_j}\in \mathbb{Z}+\dfrac{1}{2} \,\,\,\, : \, \substack{0\leq i_1<k_1\leq d_1-1 \\ 0 \leq i_2< \tilde{i_2}\leq d_2-1}.\end{equation}
\noindent and, that in which case there was a purity in the reduction of $\rho^{S}_{(r_j)_{j}}=\ket{v^{S}_{(r_j)_{j}}}\bra{v^{S}_{(r_j)_{j}}}$ is:
\begin{equation}\gamma((\rho^{S}_{(r_j)_j})_{1})=\dfrac{1}{d}(d_1+d_2-1-(d_1-1)(d_2-1))=\dfrac{2(d_1+d_2-1)-d}{d}.\end{equation}
\noindent Performing an analysis for possible values of $d_1$ and $d_2$ it was found that the only possible dimensions where maximum entanglement can exist are $d_1=d_2=1$ and $d_1=d_2=3$. The existence of negative values for purity suggests that there are no values such that the necessary condition of maximum entanglement is fulfilled.\\

In the second section of results (section \ref{multipartite}), it was shown the proposition \ref{prop 3.2.1}, which states that for a multipartite system $d=d_1\dots d_n$, the $(d_1,\dots,d_n)$-separability of the state (\ref{4.1}), is given if and only if:
\begin{equation}\label{4.5}\sum_{j=0}^{m-1} \dfrac{r_j\,((s_{(i_ld^{\prime}_{l+1}+\tilde{i}_{l+1}),j}-s_{(k_ld^{\prime}_{l+1}+\tilde{i}_{l+1}),j})+(s_{(k_l d^{\prime}_{l+1}+\tilde{k}_{l+1}),j}-s_{(i_ld^{\prime}_{l+1}+\tilde{k}_{l+1}),j}))}{n_j}\in \mathbb{Z} \,\,\,\, : \, \substack{0\leq i_l<k_l\leq d_l-1 \\ 0 \leq \tilde{i}_{l+1}< \tilde{k}_{l+1}\leq d_{l+1}^{\prime}-1\\ 1 \leq l\leq n-1},\end{equation}
\noindent where $d^{\prime}_l=d_{l}d_{l+1}\dots d_{n}$, and in this case, the reductions of $\rho^{S}_{(r_j)_{j}}$ are: 
\begin{equation}(\rho^{S}_{(r_j)_j})_{l}=\dfrac{1}{d_l}\sum_{0 \leq i_l,k_l\leq d_l-1} \prod_{j=0}^{m-1} w_j^{r_j \,(s_{i_ld_{l+1}^{\prime},j}-s_{k_l d_{l+1}^{\prime},j})}\, \ket{i_l}\bra{k_l},\end{equation}
\noindent where $d_{n+1}^{\prime}:=1$ (empty product convention).

\noindent This suggested that in such a case the state could be written as a tensor product of the following factors:
\begin{equation}\ket{v^{S}_{(r_j)}}_{l}=\dfrac{1}{\sqrt{d_l}}\sum_{0 \leq i_l\leq d_l-1} \prod_{j=0}^{m-1} w_j^{r_j \,(s_{id_{l+1}^{\prime},j}+\theta_j)}\, \ket{i_l}.\end{equation}
\noindent It was shown in the observation \ref{obs 3.2.1} that one possible value for $\theta_j$ is:
\begin{equation}\theta_j=\left(\dfrac{1-n}{n}\right) s_{0,j}.\end{equation}
\noindent It was also observed in \ref{obs 3.2.2} that another valid separability equivalence is given when:
\begin{equation}\sum_{j=0}^{m-1} \dfrac{r_j\,((s_{i_ld^{\prime}_{l+1},j}-s_{0,j})+(s_{k_{l+1}(i_r)_{r\geq l+1},j}-s_{k_{l}(i_r)_{r\geq l},j}))}{n_j}\in \mathbb{Z} \,\,\,\, : \, \substack{0\leq i_r\leq d_r-1 \\  1 \leq l\leq n},\end{equation}
\noindent where:
\begin{equation}k_{l}(i_r)_{r\geq l}:=\sum_{r=l}^{n}i_{r}d^{\prime}_{r+1},\end{equation}
\noindent and $k_{n+1}(i_r)_{r\geq n+1}:=0$ (empty sum convention).\\

Analyzing the proof of Naimark's theorem (proposition \ref{prop 2.6.1}), it was subsequently shown in the section (\ref{Cn frames}) that it is possible to implement $C_N$ frames $(m=1)$ of dimension $d$ as POVM, using the Fourier matrix $F_{dN}$ and a permutation matrix $P_{\sigma_{S}}$ via this theorem. Specifically, the proposition \ref{prop 3.3.1}, says that this construction implement the $C_N$ frame in an order $\sigma^{\prime}$ if and only if for all:
\begin{equation}d\vert\sigma_{S}(\tilde{j}N)-\sigma_S(jN)\,\,\,\,\,\, \wedge \,\,\,\,\, N\bigg\vert\left(k\left(\dfrac{\sigma_S(\tilde{j}N)-\sigma_S(jN)}{d}\right)+\sigma^{\prime}(k)(s_{\tilde{j}}-s_{j})\right)\,\,\,\, : \, \substack{0\leq j<\tilde{j}\leq d-1 \\ 0 \leq k\leq N-1}.\end{equation}
\noindent It was demonstrated the existence of a permutation that designs the frame in order taking
\( C \in \{\max\limits_{0\leq i\leq d-1} s_id, \ldots, Nd-1\} \) and \( \sigma_{S}(iN)=C-s_id \). Finally, the results were applied to implement some simple frames on quantum computers. As a POVM we designed a circuit for the case of the $C_4$ frame of dimension $d=2$ such that $s_j=j$ (section \ref{example povm}). This is done by permuting $\sigma_{S}=(46)(57)$ in order $\sigma^{\prime}(k)=k$. The permutation matrix associated with the sigma turned out to be a $CNOT_{0\to 1}$. so that the circuit obtained after implementing $F_8$ by the proposed method in \cite{nielsen2010quantum} was:

\begin{figure}[hbtp]
    \centering
    \begin{quantikz}[row sep=0.8cm]
\lstick{$\rho$}& \ctrl{1}\gategroup[2,steps=1,style={dashed,rounded
corners,fill=blue!20, inner
xsep=2pt},background,label style={label
position=below,anchor=north,yshift=-0.2cm}]{{\sc
$\textcolor{blue}{P_{\sigma_{S}}}$}}&\gate{H} \gategroup[3,steps=7,style={dashed,rounded
corners,fill=red!20, inner
xsep=2pt},background,label style={label
position=below,anchor=north,yshift=-0.2cm}]{{\sc
$\textcolor{red}{F_{8}}$}}& \gate{S}& \gate{T}&&&&\swap{2}& &\rstick{$\rho^{S}_{2i+j}$}\\
\lstick{$\ket{0}$}& \gate{X}&&\ctrl{-1}  & &\gate{H}&\gate{S}&&&&\meter{i} \\
\lstick{$\ket{0}$}&  &&&\ctrl{-2}&&\ctrl{-1}&\gate{H}&\targX{}&&\meter{j}
\end{quantikz}
    \caption{\doublespacing We show the implementation of the POVM (\ref{3.61}) using the proposition \ref{prop 3.3.1}. In this case, two ancilla qubit were used, the gate $P_{\sigma_{S}}=CNOT_{0\to 1}$, and $F_8$, whose construction is described in \cite{nielsen2010quantum}.}
    \label{fig:enter-label}
\end{figure}
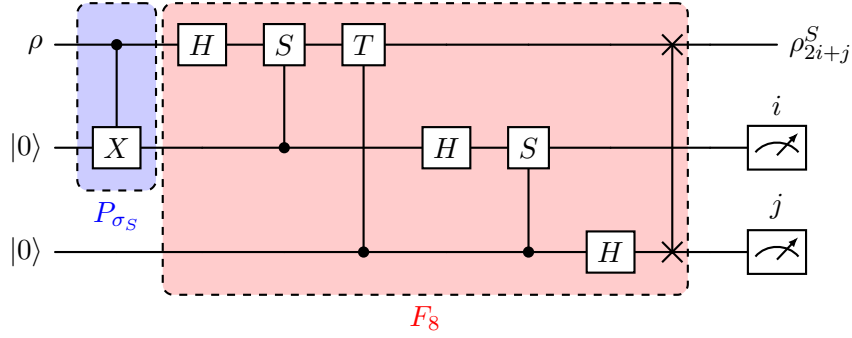

\noindent Viewed as states, the separability conditions were applied to the case of $m=1$ ($C_N$ frame), with $n=2$ (Bipartite), $d_1=d_2=2$ $(d=4)$, and $N=4$ $(S=s_i)$ (section \ref{example separability}). The separability was also illustrated for the general case of $N\geq 4$ and any S. The condition in this case was as follows $k(s_2+s_1-s_0)=k s_3$ (mod $N$) and:

\begin{equation}\ket{v^{S}_k}=(\ket{v^{S}_{k}})_{1}\otimes(\ket{v^{S}_{k}})_{2}=\dfrac{1}{\sqrt{2}}\begin{pmatrix}
    w_{N}^{k \frac{s_0}{2}}\\
    w_{N}^{k (s_{2}-\frac{s_0}{2})}
\end{pmatrix}\otimes \dfrac{1}{\sqrt{2}}\begin{pmatrix}
    w_{N}^{k \frac{s_{0}}{2}}\\
    w_{N}^{k (s_{1}-\frac{s_0}{2})}
\end{pmatrix}.\end{equation}

\section{Reflections}\label{reflections}

Throughout this thesis, tools for the implementation of harmonic frames in quantum computers were developed and illustrated with simple examples. These tools were instrumental in exploring how harmonic frames could be adapted to practical scenarios in quantum computing, offering insights into their potential applications. One notable aspect was the role of separability, which, while not directly utilized in the presented methods, holds significant promise for simplifying the construction of quantum states. Specifically, separability could facilitate the creation of less complex and computationally expensive states, as its local nature allows the design process to be broken down into several smaller, lower-dimensional components.\\

In the specific context of qubit-based quantum computers, the methods developed enabled the design of $C_{2^M}$ frames that could be implemented using native quantum gates. This achievement represents an important step forward in addressing the challenges associated with constructing POVMs, as it provides a pathway to more efficient implementations. By aligning these constructions with the native operations of quantum hardware, the developed techniques contribute to bridging the gap between theoretical designs and their practical realization.\\

Overall, while significant progress was achieved, much remains to be explored in this area. The methods and tools presented here lay a foundation for further research, yet they also highlight the complexity of the problem and the opportunities for future advancements. As discussed later, deeper analyses and broader generalizations of the approaches introduced in this work could lead to a more comprehensive understanding of harmonic frames and their applications in quantum computing.

\section{Future Research}\label{future research}

Regarding the separability conditions (\ref{4.5}), future work could focus on reducing the number of trivial conditions. In the application example, many conditions were straightforward and obvious; thus, achieving a minimum set of non-trivial conditions would be ideal. On the other hand, in terms of designing frames as POVMs, while a viable circuit was obtained for quantum computers operating with powers of $2$, this method is not feasible for other $C_N$ in qubit-based systems. This limitation arises because the Fourier matrix inherently produces roots of unity that are powers of two. For instance, implementing a cubic root would not be possible. Therefore, a remaining task is to investigate how such implementations can be realized on quantum computers. Additionally, developing a general circuit for all harmonic frames in an abstract sense would be a valuable contribution. This may be achievable through products of Fourier matrices, as suggested by the structure of the abelian group. \\

Another avenue for exploration is analyzing the computational cost of the circuits, particularly the implementation of the permutation matrices (\ref{prop 3.3.1}). It is well-known that the Toffoli gate is a universal gate for reversible Boolean circuits \cite{kitaev2002classical}, making it worthwhile to determine the minimum number of Toffoli gates required to implement the permutations forming the circuit. Furthermore, it would be beneficial to investigate whether simpler gates could be used to achieve the same functionality. However, there are intuitions that accomplishing this solely with local gates may not be feasible. \\

In conclusion, I could likely continue generating ideas all night, but I will leave it here. It has been profoundly rewarding to see how knowledge from diverse fields can converge to address a complex problem. Simply contemplating the possibilities is enough to give me goosebumps.

\bibliography{references}  

\begin{thebibliography}{10}
\providecommand{\url}[1]{#1}
\csname url@samestyle\endcsname
\providecommand{\newblock}{\relax}
\providecommand{\bibinfo}[2]{#2}
\providecommand{\BIBentrySTDinterwordspacing}{\spaceskip=0pt\relax}
\providecommand{\BIBentryALTinterwordstretchfactor}{4}
\providecommand{\BIBentryALTinterwordspacing}{\spaceskip=\fontdimen2\font plus
\BIBentryALTinterwordstretchfactor\fontdimen3\font minus
  \fontdimen4\font\relax}
\providecommand{\BIBforeignlanguage}[2]{{%
\expandafter\ifx\csname l@#1\endcsname\relax
\typeout{** WARNING: IEEEtran.bst: No hyphenation pattern has been}%
\typeout{** loaded for the language `#1'. Using the pattern for}%
\typeout{** the default language instead.}%
\else
\language=\csname l@#1\endcsname
\fi
#2}}
\providecommand{\BIBdecl}{\relax}
\BIBdecl

\bibitem{duffin1952class}
R.~J. Duffin and A.~C. Schaeffer, ``A class of nonharmonic fourier series,''
  \emph{Trans. Amer. Math. Soc.}, vol.~72, pp. 341--366, 1952.

\bibitem{gabor1946theory}
D.~Gabor, ``Theory of communication,'' \emph{Journ. IEE}, vol.~93, pp.
  429--457, 1946.

\bibitem{casazza2000art}
P.~Casazza, ``The art of frame theory,'' \emph{Taiwanese Journ. of Math.},
  vol.~4, no.~2, pp. 129--202, 2000.

\bibitem{eldar2002optimal}
Y.~Eldar and J.~G.D.~Forney, ``Optimal tight frames and quantum measurement,''
  \emph{IEEE Trans. Info. Th.}, vol.~48, no.~3, pp. 599--610, March 2002.

\bibitem{beneduci2018notes}
R.~Beneduci, ``Notes on naimark's dilation theorem,'' \emph{Journal of
  Mathematical Physics}, vol.~59, no.~3, p. 032107, 2018, departimento di
  Fisica, Universit\'a della Calabria and Istituto Nazionale di Fisica
  Nucleare, Italy.

\bibitem{benedetto2002finite}
J.~Benedetto and M.~Fickus, ``Finite normalized tight frames,'' \emph{Advances
  in Computational Mathematics}, 2002.

\bibitem{yard2024introduction}
J.~Yard, ``Introduction to quantum information processing,'' QIC 710 / CS 768 /
  PH 767 / CO 681 / AM 871, 2024, qNC 3126,
  \url{http://math.uwaterloo.ca/~jyard/qic710}.

\bibitem{renes2004symmetric}
J.~M. Renes, R.~Blume-Kohout, A.~J. Scott, and C.~M. Caves, ``Symmetric
  informationally complete quantum measurements,'' \emph{Journal of
  Mathematical Physics}, vol.~45, no.~6, pp. 2171--2180, 2004.

\bibitem{bergou2010discrimination}
J.~A. Bergou, ``Discrimination of quantum states,'' \emph{Journal of Modern
  Optics}, vol.~57, no. 3-4, pp. 160--180, 2010.

\bibitem{derka1998universal}
R.~Derka, V.~Bu{\v{z}}ek, and A.~K. Ekert, ``Universal algorithm for optimal
  estimation of quantum states from finite ensembles via realizable generalized
  measurement,'' \emph{Physical Review Letters}, vol.~80, no.~8, pp.
  1571--1575, 1998.

\bibitem{vertesi2010two}
T.~Vértesi and E.~Bene, ``Two-qubit bell inequality for which positive
  operator-valued measurements are relevant,'' \emph{Phys. Rev. A}, vol.~82, p.
  062115, 2010.

\bibitem{gomez2016device}
E.~S. Gómez \emph{et~al.}, ``Device-independent certification of a
  nonprojective qubit measurement,'' \emph{Phys. Rev. Lett.}, vol. 117, p.
  260401, 2016.

\bibitem{bennett1992quantum}
C.~H. Bennett, ``Quantum cryptography using any two nonorthogonal states,''
  \emph{Physical Review Letters}, vol.~68, no.~21, pp. 3121--3124, 1992.

\bibitem{jozsa2003entanglement}
R.~Jozsa and et~al., ``Entanglement cost of generalised measurements,''
  \emph{Quantum Information and Computation}, vol.~3, pp. 405--422, 2003.

\bibitem{steinberg2016representation}
B.~Steinberg, \emph{Representation Theory of Finite Groups: An Introductory
  Approach}.\hskip 1em plus 0.5em minus 0.4em\relax New York: Springer, 2016.

\bibitem{waldron2018group}
\BIBentryALTinterwordspacing
S.~Waldron, ``Group frames,'' 2018. [Online]. Available:
  \url{https://www.math.auckland.ac.nz/~waldron/Preprints/Group-frames/group-frames.pdf}
\BIBentrySTDinterwordspacing

\bibitem{vale2018symmetry}
\BIBentryALTinterwordspacing
R.~Vale and S.~Waldron, ``The symmetry group of a finite frame,'' 2018.
  [Online]. Available: \url{https://arxiv.org/abs/1812.04898}
\BIBentrySTDinterwordspacing

\bibitem{fuchs2015abelian}
L.~Fuchs, \emph{Abelian Groups}, ser. Graduate Texts in Mathematics.\hskip 1em
  plus 0.5em minus 0.4em\relax Berlin: Springer, 2015, vol. 189.

\bibitem{marshall2018number}
\BIBentryALTinterwordspacing
S.~Marshall and S.~Waldron, ``On the number of harmonic frames,'' \emph{ArXiv},
  2018, october 16, 2018. [Online]. Available:
  \url{https://arxiv.org/abs/1810.07415}
\BIBentrySTDinterwordspacing

\bibitem{casazza2008finite}
P.~G. Casazza and G.~Kutyniok, Eds., \emph{Finite Frames: Theory and
  Applications}.\hskip 1em plus 0.5em minus 0.4em\relax Boston: Birkhäuser,
  2008.

\bibitem{tung1985group}
W.~K. Tung, \emph{Group Theory in Physics}.\hskip 1em plus 0.5em minus
  0.4em\relax Singapore: World Scientific, 1985, vol.~1.

\bibitem{waldron2017tightframes}
S.~F.~D. Waldron, \emph{An Introduction to Finite Tight Frames}.\hskip 1em plus
  0.5em minus 0.4em\relax Springer, July 2017, draft available online.

\bibitem{piro2013fundamental}
A.~Piro, ``The fundamental theorem for finite abelian groups: A brief history
  and proof,'' \emph{Missouri Journal of Mathematical Sciences}, vol.~25,
  no.~2, pp. 88--96, 2013.

\bibitem{chien2011unitary}
T.-Y. Chien, ``On the unitary equivalence between cyclic harmonic frames,''
  \emph{Linear Algebra and its Applications}, vol. 435, no.~5, pp. 1047--1057,
  2011.

\bibitem{nielsen2010quantum}
M.~A. Nielsen and I.~L. Chuang, \emph{Quantum Computation and Quantum
  Information}, 10th~ed.\hskip 1em plus 0.5em minus 0.4em\relax Cambridge, UK:
  Cambridge University Press, 2010.

\bibitem{horodecki2024multipartite}
P.~Horodecki, Łukasz Rudnicki, and K.~Życzkowski, ``Multipartite
  entanglement,'' \emph{International Centre for Theory of Quantum
  Technologies, University of Gdansk}, 2024, dated: July 31, 2024.

\bibitem{paris2004modern}
M.~Paris, ``The modern tools of quantum mechanics: A tutorial on quantum
  states, measurements, and operations,'' \emph{Dipartimento di Fisica
  dell’Università degli Studi di Milano}, 2004.

\bibitem{dummit2004abstract}
D.~S. Dummit and R.~M. Foote, \emph{Abstract Algebra}, 3rd~ed.\hskip 1em plus
  0.5em minus 0.4em\relax Hoboken, NJ: John Wiley \& Sons, 2004.

\bibitem{treil2017linear}
S.~Treil, \emph{Linear Algebra Done Wrong}.\hskip 1em plus 0.5em minus
  0.4em\relax Brown University, 2017, available online:
  \url{https://www.math.brown.edu/~treil/papers/LADW/book.pdf}.

\bibitem{pathak2013elements}
A.~Pathak, \emph{Elements of Quantum Computation and Quantum
  Communication}.\hskip 1em plus 0.5em minus 0.4em\relax CRC Press, 2013.

\bibitem{kitaev2002classical}
A.~Y. Kitaev, A.~H. Shen, and M.~N. Vyalyi, \emph{Classical and Quantum
  Computation}, ser. Graduate Studies in Mathematics.\hskip 1em plus 0.5em
  minus 0.4em\relax Providence, RI: American Mathematical Society, 2002,
  vol.~47.

\end{thebibliography}

\end{document}